\RequirePackage{amsmath}
\documentclass[runningheads]{amsart}
\usepackage{amsaddr}
\usepackage{xcolor}
\definecolor{mgt}{RGB}{238,28,189}
\usepackage[T1]{fontenc}
\usepackage{graphicx}
\usepackage{amsmath,amssymb,amsthm}
\usepackage{xspace}
\usepackage{todonotes}
\usepackage{virginialake}
\usepackage[colorlinks]{hyperref}
\hypersetup{
    linkcolor=teal,
    filecolor=teal,
    citecolor=orange,
    urlcolor=cyan
}
\usepackage[capitalise]{cleveref}
\usetikzlibrary{arrows}
\usepackage{tikz-cd}

\begin{document}
\title[A non-wellfounded proof theory of (\texorpdfstring{$\omega$}{}-)context-free languages]{A proof theory of (\texorpdfstring{$\omega$}{}-)context-free languages, \\ {via non-wellfounded proofs}} 
\author{Anupam Das \and Abhishek De}
\address{University of Birmingham, UK}
\email{\{a.das,a.de\}@bham.ac.uk}
\date{\today}
\begin{abstract}
We investigate the proof theory of regular expressions with fixed points, construed as a notation for ($\omega$-)context-free grammars. Starting with a hypersequential system for regular expressions due to Das and Pous~\cite{DasPou17:hka}, we define its extension by least fixed points and prove the soundness and completeness of its non-wellfounded proofs for the standard language model. From here we apply proof-theoretic techniques to recover an infinitary axiomatisation of the resulting equational theory, complete for inclusions of context-free languages. Finally, we extend our syntax by greatest fixed points, now computing $\omega$-context-free languages. We show the soundness and completeness of the corresponding system using a mixture of proof-theoretic and game-theoretic techniques.
\end{abstract}
\maketitle              %
\renewcommand{\phi}{\varphi}
\renewcommand{\emptyset}{\varnothing}
\renewcommand{\epsilon}{\varepsilon}
\newcommand{\defname}[1]{\textbf{{#1}}}
\newcommand{\df}{:=}
\newcommand{\bnf}{::=}
\newcommand{\Pow}{\mathcal P}
\newcommand{\pow}[1]{\Pow(#1)}

\newcommand{\IH}{\mathit{IH}}

\newtheorem{theorem}{Theorem}
\newtheorem{lemma}[theorem]{Lemma}
\newtheorem{proposition}[theorem]{Proposition}
\newtheorem{observation}[theorem]{Observation}
\crefformat{observation}{Observation~#2#1#3}

\newtheorem{corollary}[theorem]{Corollary}
\newtheorem{fact}[theorem]{Fact}

\theoremstyle{definition}
\newtheorem{definition}[theorem]{Definition}
\newtheorem{example}[theorem]{Example}
\newtheorem{remark}[theorem]{Remark}
\newcommand{\Nat}{\mathbb{N}}
\newcommand{\Int}{\mathbb{Z}}
\newcommand{\Cal}[1]{\mathcal{#1}}
\newcommand{\SF}[1]{\mathsf{#1}}
\newcommand{\resp}{\textit{resp.}\xspace}
\newcommand{\aka}{\textit{aka}\xspace}
\newcommand{\viz}{\textit{viz.}\xspace}
\newcommand{\ie}{\textit{i.e.}\xspace}
\newcommand{\cf}{\textit{cf.}\xspace}
\newcommand{\wrt}{\textit{wrt}\xspace}
\newcommand{\mgt}[1]{\textcolor{mgt}{#1}}
\newcommand{\cyn}[1]{\textcolor{cyan}{#1}}
\newcommand{\olv}[1]{\textcolor{olive}{#1}}
\newcommand{\orange}[1]{{\color{orange}#1}}
\newcommand{\purple}[1]{{\color{purple}#1}}

\newcommand{\blue}[1]{{\color{blue}#1}}
\newcommand{\green}[1]{{\color{teal}#1}}

\newtheorem{convention}[theorem]{Convention}

\newcommand{\Alphabet}{\mathcal{A}}
\newcommand{\Var}{\mathcal{V}}

\newcommand{\NL}{\mathbf{NL}}
\newcommand{\Ptime}{\mathbf{P}}
\newcommand{\Pspace}{\mathbf{PSPACE}}
\newcommand{\poly}{\mathrm{poly}}

\newcommand{\rk}[1]{|\hspace{-.15em}|#1|\hspace{-.15em}|}
\newcommand{\rkapprx}[2]{\rk{#1}^{#2}}

\newcommand{\lred}[1]{\rightarrow_{#1}}

\newcommand{\proves}{\vdash}

\newcommand{\nwfproves}{\proves^{\infty}}
\newcommand{\ImAnc}[1]{\mathrm{IA}_{#1}}
\newcommand{\imanc}[3]{\ImAnc{#1}(#2,#3)}

\newcommand{\Lang}{\mathcal L}
\newcommand{\lang}[1]{\Lang(#1)}
\newcommand{\Inf}[1]{\SF{Inf}(#1)}
\newcommand{\cons}{::}

\newcommand{\wLang}{\Lang}
\newcommand{\wlang}[1]{\wLang(#1)}

\newcommand{\FV}{\mathrm{FV}}
\newcommand{\fv}[1]{\FV(#1)}

\newcommand{\fint}[1]{\lceil #1 \rceil}

\newcommand{\FL}{\mathrm{FL}}
\newcommand{\fl}[1]{\FL(#1)}
\newcommand{\eqfl}{=_\FL}
\newcommand{\lefl}{<_\FL}
\newcommand{\leqfl}{\leq_\FL}
\newcommand{\geqfl}{\geq_\FL}

\newcommand{\subform}{\sqsubseteq}
\newcommand{\supform}{\sqsupseteq}
\newcommand{\dle}{\prec}
\newcommand{\dleq}{\preceq}
\newcommand{\dge}{\succ}
\newcommand{\dgeq}{\succeq}

\newcommand{\struct}{\mathfrak L}
\newcommand{\interp}[2]{#2^{#1}}

\newcommand{\anbn}{\{a^nb^n\}_n}
\newcommand{\anbnapprx}[1]{\{a^kb^k\}_{k< #1}}
\newcommand{\asbs}{(a^*b^*)}
\newcommand{\dyck}[1]{\mathrm{Dyck}_{#1}}

\newcommand{\es}[1]{\mathcal{#1}}
\newcommand{\ines}[1]{\mathcal{#1}^\geq}

\newcommand{\esEX}[3]{\es {#1}_{#2}(#3)}
\newcommand{\inesEX}[3]{\ines {#1}_{#2}(#3)}

\newcommand{\canes}[1]{\mathcal E_{#1}}
\newcommand{\meetes}[1]{\mathcal M_{#1}}
\newcommand{\meetvar}[2]{X_{#1 \cap #2}}

\newcommand{\rleq}{\lesssim}
\newcommand{\rgeq}{\gtrsim}
\newcommand{\req}{\approx}

\newcommand{\seqar}{\rightarrow}

\newcommand{\id}{\mathsf{init}}

\newcommand{\K}{\mathsf{k}}
\newcommand{\kk}[1]{\K_{#1}}
\newcommand{\kl}[1]{\kk{#1}^l}
\newcommand{\kr}[1]{\kk{#1}^r}
\newcommand{\wk}{\rr{\mathsf{w}}}
\newcommand{\cntr}{\mathsf{c}}
\newcommand{\cut}{\mathsf{cut}}

\newcommand{\lr}[1]{#1\text{-}l}
\newcommand{\rr}[1]{#1\text{-}r}
\newcommand{\func}{\SF{func}}

\newcommand{\KA}{\mathsf{KA}}
\newcommand{\lhKA}{\ell\KA}
\newcommand{\HKA}{\mathsf{HKA}}
\newcommand{\lHKA}{\ell\HKA}
\newcommand{\muHKA}{\mu\HKA}
\newcommand{\lmuHKA}{\mu\lHKA}
\newcommand{\lmuHKAnwf}{\lmuHKA^{\infty}}

\newcommand{\muHKAw}{\muHKA_{\omega}}

\newcommand{\muHKAnwf}{\muHKA^\infty}

\newcommand{\munuHKA}{\mu\nu\HKA}
\newcommand{\munuHKAnwf}{\munuHKA^{\infty}}

\newcommand{\lmunuHKA}{\mu\nu\lHKA}
\newcommand{\lmunuHKAnwf}{\lmunuHKA^{\infty}}

\newcommand{\CFA}{\mu\mathsf{CA}}

\newcommand{\RLA}{\ensuremath{\mathsf{RLA}}\xspace}
\newcommand{\RLAhat}{\widehat\RLA}

\newcommand{\nRLA}{\nu\RLA}
\newcommand{\nRLAhat}{\nu\RLAhat}

\newcommand{\LRLAhat}{\mathsf L\RLAhat}
\newcommand{\CRLA}{\ensuremath{\mathsf{CRLA}}\xspace}
\newcommand{\LRAind}{\ensuremath{\mathsf{LRA}^{\SF{ind}}}\xspace}

\newcommand{\nLRLAhat}{\nu\LRLAhat}
\newcommand{\nCRLA}{\nu\CRLA}

\newcommand{\infrule}{\mathsf{r}}
\newcommand{\sequent}{\mathcal{S}}

\newcommand{\eloise}{$\exists$l\"{o}ise\xspace}
\newcommand{\abelard}{$\forall$belard\xspace}
\newcommand{\prover}{\mathbf{P}}
\newcommand{\denier}{\mathbf{D}}

\newcommand{\provstrat}{\mathfrak p}
\newcommand{\denstrat}{\mathfrak d}

\newcommand{\reg}{\SF{REG}}
\newcommand{\oreg}{\SF{REG}^\omega}
\newcommand{\CF}{\SF{CF}}
\newcommand{\oCF}{\SF{CF}^\omega}

\newcommand{\lv}[1]{\overline{\vec{#1}}}

\newcommand{\gram}{\Cal G}

\newcommand{\dec}[1]{\textcolor{teal}{\scriptsize\textbf{#1}}} 
\section{Introduction}

The characterisation of context-free languages (CFLs) as the least solutions of algebraic inequalities, sometimes known as the \emph{ALGOL-like theorem}, is a folklore result attributed to several luminaries of formal language theory including Ginsburg and Rice~\cite{GinsburgRice62}, Schutzenberger~\cite{Schutzenberger63}, and Gruska~\cite{Gruska71}.
This induces a syntax for CFLs by adding least fixed point operators to regular expressions, as first noted by Salomaa~\cite{Salomaa73}. Lei\ss~\cite{Leiss92:ka-with-rec} called these constructs ``$\mu$-expressions'' and defined an algebraic theory over them by appropriately extending Kleene algebras, which work over regular expressions.
Notable recent developments include a generalisation of Antimirov's partial derivatives to $\mu$-expressions~\cite{Thiemann17} and criteria for identifying $\mu$-expressions that can be parsed unambiguously \cite{KrishnaswamiYallop19}.

Establishing axiomatisations and proof systems for classes of formal languages has been a difficult challenge. 
Many \emph{theories} of regular expressions, such as Kleene algebras ($\KA$) were proposed in the late 20\textsuperscript{th} century (see, e.g., e.g.~\cite{Conway71book,Kleene56,Kozen94:completeness-ka}). 
The completeness of $\KA$ for the (equational) theory of regular languages, due to Kozen \cite{Kozen94:completeness-ka} and Krob \cite{Krob91:ka-completeness} independently, is a celebrated result that has led to several extensions and refinements, e.g.~\cite{KozSmi97:kat-completeness,KozSil12:left-handed-completeness,CraLauStr15:omega-regular-algebras,KozSil20:left-handed-completeness}.
More recently the proof theory of $\KA$ has been studied via \emph{infinitary} systems. On one hand,~\cite{Palka07} proposed an \emph{$\omega$-branching} sequent calculus and on the other hand~\cite{DasPou17:hka,DasDouPou18:lka-completeness,HazKup22:transfin-HKA} have studied \emph{cyclic} `hypersequential' calculi. 

Inclusion of CFLs is $\Pi^0_1$-complete, so any recursive (hence also cyclic) axiomatisation must necessarily be incomplete. 
Nonetheless various
theories of $\mu$-expressions have been extensively studied, in particular \emph{Chomsky algebras} and \emph{$\mu$-semirings}~\cite{EsikLeiss02,EsikLeiss05:alg-comp-semirings,Leiss16:matrices-over-mu-cont-chom-alg,LeissHopkins18}, giving rise to a rich algebraic theory. 
Indeed Grathwohl, Henglein, and Kozen~\cite{GraHenKoz13:inf-ax-cf-langs} have  given a complete (but infinitary) axiomatisation of the equational theory of $\mu$-expressions, by extending these algebraic theories with a \emph{continuity} principle for least fixed points.

\begin{figure}[t]
    \centering
    \begin{tikzpicture}
    \draw [fill=violet!20,violet!20] (-2.5,-2) rectangle (2.5,3.5);
    \draw [fill=orange!20,orange!20] (2.5,-2) rectangle (7.5,3.5);
        \node (A) at (0,3) {$\CFA$};
        \node (B) at (0,1) {$\muHKAnwf$};
        \node (C) at (5,2) {context-free languages};
        \node (D) at (0,-1) {$\lmunuHKAnwf$};
        \node (E) at (5,-1) {$\omega$-context-free languages};
        \draw[->] (A) to[bend left] node[below] {\scriptsize\textcolor{gray}{ soundness, e.g.\ \textbf{\scriptsize\cite{GraHenKoz13:inf-ax-cf-langs}}}} (C);
        \draw[->,gray] (B) to[bend left] node[below] {\textcolor{teal}{\scriptsize\textbf{\cref{thm:soundness-muHKA}}}} (C);
        \draw[<-] (B) to[bend right] node[above] {\textcolor{teal}{\scriptsize\textbf{\cref{thm:completeness}}}} (C);
        \draw[->] (D) to[bend left] node[below] {\textcolor{teal}{\scriptsize\textbf{\cref{thm:munuHKA-soundness}}}} (E);
        \draw[<-] (D) to[bend right] node[above] {\textcolor{teal}{\scriptsize\textbf{\cref{thm:munuHKA-completeness}}}} (E);
        \draw[dashed] (-2.5,0.2) -- (7.5,0.2);
        \draw[right hook->,dashed] (C) to (E);
        \draw[->] (B) to (A);
        \node[label={[align=right]\dec{\cref{thm:muHKAnwf-to-muHKAw}}\\\dec{\cref{prop:soundness-of-muHKAw-for-CFA}}}] at (-1.2,1.5) {};
        \draw[right hook->,dashed] (B) to (D);
    \node[above right] at (-2.5,-2) {\textcolor{violet}{\small \textit{Proof systems}}};
    \node[above left] at (7.5,-2) {\textcolor{orange}{\small \textit{Language models}}};
    \end{tikzpicture}
    \caption{\textnormal{Summary of our main contributions. Each arrow $\rightarrow$ denotes an inclusion of equational theories, over an appropriate language of $\mu$-expressions. The gray arrow, \cref{thm:soundness-muHKA}, is also a consequence of the remaining black ones.}
    }
    \label{fig:summary}
\end{figure}
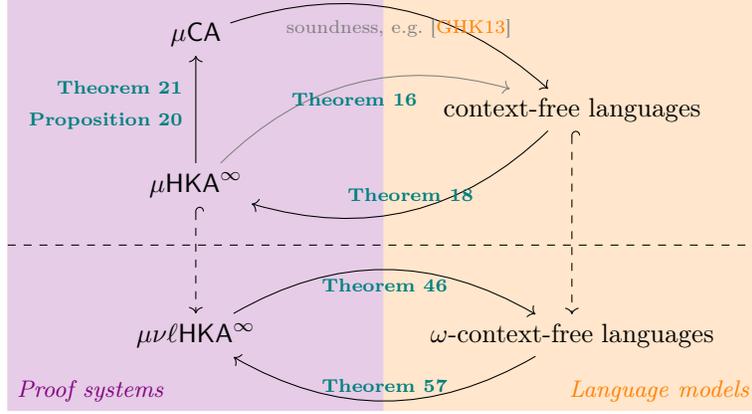

\smallskip

\noindent
\textbf{Contributions.} In this paper, we propose a \emph{non-wellfounded} system $\muHKAnwf$ for $\mu$-expressions. 
It can be seen as an extension of the cyclic system of~\cite{DasPou17:hka} for regular expressions. 
Our first main contribution is the adequacy of this system for CFLs: $\muHKAnwf$ proves $e=f$ just if the CFLs computed by $e$ and $f$, $\lang e$ and $\lang f$ respectively, are the same. 
We use this result to obtain alternative proof of completeness of the infinitary axiomatisation $\CFA$ of~\cite{GraHenKoz13:inf-ax-cf-langs}, comprising our second main result. 
Our method is inspired by previous techniques in non-wellfounded proof-theory,  namely \cite{Studer08,DDS23}, employing `projections' to translate non-wellfounded proofs to wellfounded ones.
Our result is actually somewhat stronger than that of \cite{GraHenKoz13:inf-ax-cf-langs}, since our wellfounded proofs are furthermore \emph{cut-free}.

Finally we develop an extension $\lmunuHKA$ of (leftmost) $\muHKA$ by adding \emph{greatest} fixed points, $\nu$, for which $\lang \cdot$ extends to a model of \emph{$\omega$-context-free languages}. Our third main contribution is the soundness and completeness of $\lmunuHKA$ for $\lang \cdot$. 
Compared to $\muHKA$, the difficulty for metalogical reasoning here is to control interleavings of $\mu$ and $\nu$, both for soundness argument and in controlling proof search for completeness.
To this end we employ \emph{game theoretic} techniques to characterise word membership and control proof search. 
\smallskip

All our main results are summarised in \cref{fig:summary}. 

\subsection*{Acknowledgements}
This work was supported by a UKRI Future Leaders Fellowship, ‘Structure vs Invariants in Proofs’, project reference MR/S035540/1. The authors are grateful to anonymous reviewers for their helpful comments (in particular, leading to~\Cref{ex:r1-example}) and for pointing us to relevant literature such as \cite{Lange02:model-checking-fp-logic+chop,LoedMadhSerr04:vp-games,MHHO05}.  %
\section{A syntax for context-free grammars}
\label{sec:cfg}

Throughout this work we make use of a finite set $\Alphabet$ (the \defname{alphabet}) of \defname{letters}, written $a,b,\dots$, and a countable set $\Var$ of \defname{variables}, written $X,Y,\dots$.
When speaking about context-free grammars (CFGs), we always assume non-terminals are from $\Var $ and the terminals are from $\Alphabet$.

We define \defname{($\mu$-)expressions}, written $e,f,$ etc., by:
\begin{equation}
\label{eq:grammar-of-mu-exprs}
e,f,\dots \quad \bnf 
0 \mid 
1  \mid  X  \mid  a \mid 
e+ f \mid  e\cdot f \mid \mu X e
\end{equation}
We usually simply write $ef$ instead of $e\cdot f$. $\mu$ is considered a variable binder, with the \emph{free variables} $\fv e $ of an expression $e$ defined as expected:

\begin{definition}[Free variables]
The set of \defname{free variables} of an expression $e$, written $\fv e$, is defined by:
\begin{itemize}
    \item $\fv 0 \df \emptyset$
    \item $\fv 1 \df \emptyset$
    \item $\fv X \df \{X\}$
    \item $\fv {e+f} \df \fv e \cup \fv f$
    \item $\fv {ae} \df \fv e$
    \item $\fv {\mu X e} \df \fv e \setminus \{X\}$
\end{itemize}
\end{definition}

We sometimes refer to expressions as \emph{formulas}, and write $\subform$ for the subformula relation.

$\mu$-expressions compute languages of finite words in the expected way:
\begin{definition}
    [Language semantics]
    \label{def:reg-lang-model}
    Let us temporarily expand the syntax of expressions to include each language $A\subseteq \Alphabet^*$ as a constant symbol.
We interpret each closed expression (of this expanded language) as a subset of $ \Alphabet^*$ as follows:

\smallskip

\begin{minipage}{.35\textwidth}
    \begin{itemize}
    \item $\lang 0 \df \emptyset$
    \item $\lang 1 \df \{\epsilon\}$
    \item $\lang a \df \{a\}$
    \item $\lang A \df A$
\end{itemize}
\end{minipage}
\begin{minipage}{.6\textwidth}
    \begin{itemize}
    \item $\lang {e+f} \df \lang e \cup \lang f$
    \item $\lang {ef} \df \{ vw : v \in \lang e, w \in \lang f\}$
    \item $\lang {\mu X e(X)} \df \bigcap \{ A \supseteq \lang {e(A)} \} $
\end{itemize}
\end{minipage}

\end{definition}
Note that all the operators of our syntax correspond to monotone operations on $\pow{\Alphabet^*}$, with respect to $\subseteq$. Thus $\lang {\mu X e(X)}$ is just the least fixed point of the operation $A \mapsto \lang{e(A)}$, by the Knaster-Tarski fixed point theorem.

The \defname{productive} expressions, written $p,q$ etc.\ are generated by:
\begin{equation}
\label{eq:prod-exprs-mu-only}
    p,q,\dots \quad \bnf \quad a \quad \mid \quad p + q \quad \mid \quad  p\cdot e \quad \mid \quad e\cdot p \quad \mid \quad \mu X p
\end{equation}
We say that an expression is \defname{guarded} if each variable occurrence occurs free in a productive subexpression.
\defname{Left-productive} and \defname{left-guarded} are defined in the same way, only omitting the clause $e\cdot p$ in the grammar above.
For convenience of exposition we shall employ the following convention throughout:
\begin{convention}
    Henceforth we assume all expressions are guarded.
\end{convention}

\begin{remark}
    [Why only guarded expressions?]
    There are several reasons for employing this convention. 
Most importantly, left-guardedness will be required for our treatment of $\omega$-words later via greatest fixed points, where grammars naturally parse from the left via \emph{leftmost derivations}. 
In the current setting, over finite words with only least fixed points, it makes little difference whether we use only guarded expressions or not, nor whether we guard from the left or right.
However our convention does simplify some proofs and change some statements; we will comment on such peculiarities when they are important.
\end{remark}

\begin{example}
[Empty language]
\label{ex:empty-lang}
    In the semantics above, note that the empty language $\emptyset$ is computed by several expressions, not only $0$ but also $\mu X X$ and $\mu X (aX)$. 
    Note that whle the former is unguarded the latter is (left-)guarded.
    In this sense the inclusion of $0$ is somewhat `syntactic sugar', but it will facilitate some of our later development.
\end{example}

\begin{example}
    [Kleene star and universal language]
    \label{ex:kleene-star-and-univ}
    For any expression $e$ we can compute its Kleene star $e^* \df \mu X (1 + eX)$ or $e^* \df \mu X (1 + Xe)$. 
    These definitions are guarded just when $e $ is productive.
     Now, note that we also have not included a symbol $\top$ for the universal language $\Alphabet^*$.
     We can compute this by the expression $\left(\sum \Alphabet\right)^*$, which is guarded as $\sum \Alphabet$ is productive.
\end{example}

It is well-known that $\mu$-expressions compute just the context-free (CF) languages~\cite{GinsburgRice62,Schutzenberger63,Gruska71}.
In fact this holds even under the restriction to left-guarded expressions, by simulating the \emph{Greibach normal form}:
\begin{theorem}
[Adequacy, see, e.g., \cite{EsikLeiss02,EsikLeiss05:alg-comp-semirings}]
\label{thm:cf-equiv-mu-definable}
    $L$ is context-free (and $\epsilon \notin L$) $\iff$ $L=\lang e$ for some $e$ left-guarded (and left-productive, respectively).
\end{theorem}

While this argument is known, it is pertinent to recall it as we are working with only guarded expressions, and as some of the intermediate concepts will be useful to us later.
First we will need to define a notion of subformula peculiar to fixed point expressions:

\begin{definition}[Fisher-Ladner (FL)]
\label{def:fl}
The \defname{Fischer-Ladner (FL) closure} of an expression $e$, written $\fl e$, is the smallest set of expressions closed under closed subformulas and, whenever $\mu X f(X)\in \fl e $ then $f(\mu X f(X))\in \fl e$. 
\end{definition}

It is well-known that $\fl e $ is finite, and in fact has size linear in that of $e$.
Now the $\impliedby$ direction of the Adequacy theorem, \cref{thm:cf-equiv-mu-definable}, can be proved by construing $\fl e $ as the non-terminals of an appropriate CFG.
Formally, this is broken up into the following two Propositions.

\begin{proposition}
\label{prop:exprs-are-cf}
    $\lang e $ is context-free, for any closed expression $e$.
\end{proposition}

\begin{proof}
[Proof of \cref{prop:exprs-are-cf}]
    We construct a CF grammar with nonterminals $X_f$ for each  $f \in \fl e$, starting nonterminal $X_e$, and all productions of form:
   \begin{equation}
       \label{eq:can-grammar-mu-only}
        \begin{array}{r@{\quad \rightarrow \quad }l}
         X_1 & \epsilon \\
         X_a & a \\
         X_{f+g} & X_f \quad | \quad X_g \\
         X_{fg} & X_fX_g \\
         X_{\mu X f(X)}& X_{f(\mu X f(X))} \qedhere
    \end{array}
   \end{equation}
\end{proof}

For the $\implies$ direction of \cref{thm:cf-equiv-mu-definable}, since we assume only guarded expressions, we work with grammars in \emph{Greibach normal form} (GNF).
Recall that a GNF grammar is one for which each production has form $X \to a \vec X$ or $X\to \epsilon$, for $X,\vec X$ non-terminals and $a\in \Alphabet$.
It is well-known that such grammars exhaust all CF languages~\cite{HopcroftUllmanMotwani06}.
Thus we obtain a slightly stronger result:

\begin{proposition}
[see, e.g., \cite{EsikLeiss02,EsikLeiss05:alg-comp-semirings}]
\label{prop:cf-are-mu-definable}
    $L$ is context-free $\implies$ $L= \lang e$ for some left-guarded expression $e$.
    Moreover if $\epsilon \notin L$ then $e$ is left-productive. 
\end{proposition}
\begin{proof}
[Proof sketch]
We expand the statement to grammars where each non-terminal has a unique production whose RHS is an arbitrary left-guarded $\mu$-expression, i.e.\ of the form $\{X_i \to e_i(\vec X)\}_{i<n}$ for non-terminals $\vec X = X_0, \dots, X_{n-1}$.
Note that this exhausts all context-free languages by (a) assuming Greibach normal form for left-guardedness; and (b) using $+$ to combine multiple productions from the same non-terminal.
From here we proceed by induction on $n$, the number of non-terminals, using \emph{Beki\'c's Theorem} for resolving equational systems.
Namely from $\{X_i \to e_i(\vec X, X_n)\}_{i\leq n}$
we set $e_n'(\vec X)\df \mu X_n e_n (\vec X,X_n)$ and first find solutions $\vec f$ to the grammar $\{X_i \to e_i(\vec X, e_n' (\vec X))\}_{i<n}$, by inductive hypothesis. 
Now we set the solution for $X_n$ to be $f_n \df e_n'(\vec f)$.
Note that, since we did not introduce identities, the solutions contain $1$ just if there is an $\epsilon$ production for some non-terminal, by Greibach normal form.
\end{proof}

\begin{example}
\label{eg:mu-exprs-dyck-anbn-asbs}
Consider the left-guarded expressions $\dyck 1 \df \mu X(1 + \langle X \rangle X) $ and $\anbn \df \mu X(1+ a X b) $.
As suggested, $\dyck 1 $ indeed computes the language of well-bracketed words over alphabet $\{\langle, \rangle\}$,
whereas $\anbn$ computes the set of words $\vec a\vec b$ with $|\vec a| = |\vec b|$.
We can also write $\asbs \df \mu X(1 + aX + Xb)$, which is guarded but not left-guarded.
However, if we define Kleene $*$ as in \cref{ex:kleene-star-and-univ}, then we can write $a^*$ and $b^*$ as left-guarded expressions and then take their product for an alternative representation of $\asbs$.
Note that the empty language $\emptyset$ is computed by the left-guarded expression $\mu X (aX)$, cf.~\cref{ex:empty-lang}.
\end{example} %
\section{A non-wellfounded proof system}
\label{sec:muHKA}

In this section we extend a calculus $\HKA$ from \cite{DasPou17:hka} for regular expressions to all $\mu$-expressions, and prove soundness and completeness of its non-wellfounded proofs for the language model $\lang \cdot$.
We shall apply this result in the next section to deduce completeness of an infinitary axiomatisation for $\lang \cdot$, before considering the extension to \emph{greatest} fixed points later.

\smallskip

A \defname{hypersequent} has the form $\Gamma\seqar S$ where $\Gamma$ (the LHS) is a list of expressions (a \defname{cedent}) and $S$ (the RHS) is a set of such lists. 
We interpret lists by the product of their elements, and sets by the sum of their elements.
Thus we extend our notation for language semantics by $\lang \Gamma \df \lang {\prod\Gamma}$ and $\lang S \df \bigcup\limits_{\Gamma \in S } \lang \Gamma$.

\begin{figure}[t]
    \[
    \begin{array}{l@{\quad}c}
        \text{\parbox{2cm}{\textbf{Non-logical rules:}}}
    &
\vlinf{\id}{}{\seqar [\, ]}{}
\qquad
\vlinf{\wk}{}{\orange\Gamma \seqar \purple S,[\Delta]}{\orange\Gamma \seqar \purple S}
\qquad
\vlinf{\kl a }{}{a, \orange\Gamma \seqar a\purple{S}}{\orange\Gamma \seqar \purple S}
\qquad
\vlinf{\kr a}{}{\orange \Gamma,a \seqar \purple S a}{\orange\Gamma \seqar \purple S}  
\\
\noalign{\smallskip}
     \text{\parbox{1.8cm}{\textbf{Left logical rules:}}}    
     & 
     \begin{array}{c}
     \vlinf{\lr 0 }{}{\orange \Gamma, \mgt 0,\olv {\Gamma'} \seqar \purple S}{}
\qquad
\vlinf{\lr 1 }{}{\orange \Gamma, \mgt 1, \olv {\Gamma'} \seqar \purple S}{\orange \Gamma, \olv \Gamma' \seqar \purple S}
\qquad
\vlinf{\lr \cdot}{}{\orange\Gamma, \mgt{ef}, \olv{\Gamma'} \seqar \purple S}{\orange \Gamma, \mgt e,\mgt f,\olv{\Gamma'} \seqar \purple S} 
\\
\noalign{\smallskip}
\vliinf{\lr +}{}{\orange\Gamma, \mgt{e + f} , \olv{\Gamma'} \seqar \purple S}{\orange \Gamma, \mgt e, \olv{\Gamma'} \seqar \purple S}{\orange \Gamma, \mgt f, \olv{\Gamma'} \seqar \purple S}
\qquad
\vlinf{\lr \mu}{}{\orange \Gamma, \mgt{\mu X e(X)}, \olv{\Gamma'} \seqar \purple S}{\orange \Gamma, \mgt{e(\mu X e(X))}, \olv{\Gamma'} \seqar \purple S}
\end{array}
\\
\noalign{\smallskip}
\text{\parbox{2.1cm}{\textbf{Right logical rules:}}}
&
\begin{array}{c}
     \vlinf{\rr 0}{}{\orange \Gamma \seqar \purple S, [\blue \Delta, \mgt 0,\green{\Delta'}]}{\orange \Gamma \seqar \purple S}
\qquad
\vlinf{\rr 1}{}{\orange \Gamma \seqar \purple S, [\blue \Delta, \mgt 1,\green{\Delta'}]}{\orange \Gamma \seqar \purple S, [\blue \Delta, \green{\Delta'}]}
\qquad
\vlinf{\rr \cdot}{}{\orange \Gamma \seqar \purple S, [\blue \Delta, \mgt{ef}, \green{\Delta'}]}{\orange \Gamma \seqar \purple S, [\blue \Delta, \mgt e,\mgt f,\green{\Delta}]}
\\
\noalign{\smallskip}
\vlinf{\rr +}{}{\orange\Gamma \seqar \purple S, [\blue \Delta, \mgt{e+f},\green{\Delta'}]}{\orange \Gamma \seqar \purple S, [\blue \Delta, \mgt e,\green{\Delta'}],[\blue \Delta, \mgt f,\green{\Delta'}]}
\qquad
\vlinf{\rr \mu}{}{\orange \Gamma \seqar \purple S, [\blue  \Delta, \mgt{\mu X e(X)},\green {\Delta'}]}{\orange\Gamma \seqar \purple S, [\blue \Delta, \mgt{e(\mu X e(X))}, \green {\Delta'}]}
\end{array}
    \end{array}
\]
\caption{\textnormal{Rules of the system $\muHKA$.}}
\label{fig:muHKA-rules}
\end{figure}

The system $\muHKA$ is given by the rules in \cref{fig:muHKA-rules}.
Here we use commas to delimit elements of a list or set and square brackets $[, ]$ to delimit lists in a set.
In the $\K$ rules, we write $aS \df \{[a,\Gamma] : \Gamma \in S\}$ and $Sa \df \{[\Gamma,a]:\Gamma \in S\}$.

For each inference step, as typeset in \cref{fig:muHKA-rules}, the \defname{principal} formula is the distinguished \mgt{magenta} formula occurrence in the lower sequent, while any distinguished \mgt{magenta} formula occurrences in upper sequents are \defname{auxiliary}. 
(Other colours may be safely ignored for now).

Our system differs from the original presentation of $\HKA$ in \cite{DasPou17:hka} as (a) we have general fixed point rules, not just for the Kleene $*$; and (b) we have included both left and right versions of the $\K$ rule, for symmetry. 
We extend the corresponding notions of non-wellfounded proof appropriately:

\begin{definition}
    [Non-wellfounded proofs]
        A \defname{preproof} (of $\muHKA$) is generated \emph{coinductively} from the rules of $\muHKA$ i.e.\ it is a possibly infinite tree of sequents (of height $\leq \omega$) generated by the rules of $\muHKA$.
    A preproof is \defname{regular} or \defname{cyclic} if it has only finitely many distinct subproofs.
    An infinite branch of a preproof is \defname{progressing} if it has infinitely many $\lr \mu$ steps.
    A preproof is progressing, or a \defname{$\infty$-proof}, if all its infinite branches are progressing.
    We write $\muHKA \nwfproves \Gamma \seqar S$ if $\Gamma \seqar S$ has a $\infty$-proof in $\muHKA$, and sometimes write $\muHKAnwf$ for the class of $\infty$-proofs of $\muHKA$.
\end{definition}

Note that our progress condition on preproofs is equivalent to simply checking that every infinite branch has infinitely many left-logical or $\K$ steps, as $\lr\mu$ is the only rule among these that does not decrease the size of the LHS.
This is simpler than usual conditions from non-wellfounded proof theory, as we do not have any alternations between the least and greatest fixed points.
Indeed we shall require a more complex criterion later when dealing with $\omega$-languages.
Note that, as regular preproofs may be written naturally as finite graphs, checking progressiveness for them is efficiently decidable (even in $\NL$, see e.g.\ \cite{DasPou17:hka,CurDas22:cic}).

The need for such a complex hypersequential line structure is justified in \cite{DasPou17:hka} by the desideratum of \emph{regular} completeness for the theory of regular expressions: intuitionistic `Lambek-like' systems, cf.~e.g.~\cite{Jip04:semirings-res-kls,Palka07,DasPous18:lka-pt} are incomplete (wrt regular cut-free proofs).
The complexity of the RHS of sequents in $\HKA$ is justified by consideration of proof search for, say, $a^* \seqar (aa)^* + a(aa)^*$ and $(a+b)^* \seqar a^*(ba^*)^*$, requiring reasoning under sums and products, respectively.

In our extended system we actually gain \emph{more} regular proofs of inclusions between context-free languages.
For instance:

\begin{figure}
    \[
    \footnotesize
    \vlderivation{
\vlin{\lr \mu, \rr\mu}{\bullet}{\mgt\anbn \seqar [\asbs]}{
\vliin{\lr +}{}{\mgt{1 + a\anbn b} \seqar [ 1 + a\asbs + \asbs b ]}{
    \vlin{\wk,\rr +}{}{1  \seqar [ 1 + a\asbs + \asbs b ]}{
    \vlin{\lr 1 , \rr 1}{}{1 \seqar [1]}{
    \vlin{\id}{}{\seqar [\, ]}{\vlhy{}}
    }
    }
}{
    \vlin{\wk,\rr+}{}{\mgt{a\anbn b} \seqar [ 1 + a\asbs + \asbs b ]}{
    \vlin{\lr \cdot, \rr \cdot}{}{\mgt{a\anbn b} \seqar [  a\asbs  ]}{
    \vlin{\kl a}{}{a,\mgt{\anbn b} \seqar [  a,\asbs  ]}{
    \vlin{\rr \mu,\rr +,\wk}{}{\mgt{\anbn b} \seqar [  \asbs  ]}{
    \vlin{\lr \cdot, \rr \cdot}{}{\mgt{\anbn b} \seqar [  \asbs b  ]}{
    \vlin{\kr b}{}{\mgt\anbn ,b \seqar [  \asbs ,b  ]}{
    \vlin{\lr\mu,\rr\mu}{\bullet}{\mgt\anbn \seqar \asbs}{\vlhy{\vdots}}
    }
    }
    }
    }
    }
    }
}
}
}
\]
\caption{A regular $\infty$-proof $R$ of $\anbn \seqar [\asbs]$.}
\label{eq:anbn-in-asbs}
\end{figure}

\begin{example}
\label{ex:anbn-in-asbs}
Recall the guarded expressions $\anbn$ and $\asbs$ from \cref{eg:mu-exprs-dyck-anbn-asbs}. 
We have the regular $\infty$-proof $R$ in \cref{eq:anbn-in-asbs} of $\anbn \seqar [\asbs]$,
where $\bullet$ marks roots of identical subproofs. 
Note that indeed the only infinite branch, looping on $\bullet$, has infinitely many $\lr\mu$ steps.
\end{example}

\begin{remark}
    [Impossibility of general regular completeness]
    At this juncture let us make an important point: it is impossible to have any (sound) recursively enumerable system, let alone regular cut-free proofs, complete for context-free inclusions, since this problem is $\Pi^0_1$-complete (see e.g.~\cite{HopcroftUllmanMotwani06}).
    In this sense examples of regular proofs are somewhat coincidental.
\end{remark}

It is not hard to see that each rule of $\muHKA$ is sound for language semantics:
\begin{lemma}[Local soundness]
\label{lemma:local-soundness}
For each inference step,
\begin{equation}
\label{eq:inf-step}
    \vliiinf{\infrule}{}{\Gamma \seqar S}{\Gamma_0 \seqar S_0}{\cdots}{\Gamma_{k-1}\seqar S_{k-1}}
\end{equation}
 for some $k\leq 2$, we have: $\forall i<k \, \lang {\Gamma_i} \subseteq \lang {S_i} \implies \lang \Gamma \subseteq \lang S$.
\end{lemma}
Consequently finite proofs are also sound, by induction on their structure.
For non-wellfounded proofs, we must employ a less constructive argument, typical of non-wellfounded proof theory:

\begin{theorem}
    [Soundness]
    \label{thm:soundness-muHKA}
        $\muHKA \nwfproves \Gamma\seqar S \implies \lang \Gamma \subseteq \lang S$.
\end{theorem}

\begin{proof}
[Proof of \cref{thm:soundness-muHKA}]
For each sequent $\sequent = \Gamma \seqar S$, define $n_{\sequent} \in \Nat\cup\{\infty\}$ for the least length of word $w\in \lang \Gamma \setminus \lang S$ (if there is no such word then $n_{\sequent} = \infty$).
Now suppose, for contradiction, that $P$ is a $\infty$-proof of $\sequent = \Gamma \seqar S$, but $\lang \Gamma \setminus \lang S \neq \emptyset$, and so $n_\sequent \in \Nat$.
By (the contraposition of) \cref{lemma:local-soundness} we may continually choose invalid premisses of rules to build an infinite branch $(\sequent_i = \Gamma_i \seqar S_i)_{i<\omega}$ s.t.\ $\lang{\Gamma_i}\setminus \lang{S_i} \neq \emptyset$ for all $i<\omega$.
Moreover, we can guarantee that the sequence $(n_{\sequent_i})_{i<\omega}$ is monotone non-increasing.
For this note that for each inference step $\vliiinf{\infrule}{}{\sequent}{\sequent_0}{\cdots }{\sequent_{k-1}}$ of form as in \eqref{eq:inf-step} we have:
\begin{enumerate}
    \item\label{item:k-step-strict-decrease} If $\infrule$ is a $\K$ step, then $n_\sequent > n_{\sequent_0}$; and,
    \item Otherwise there is some $i<k$ with $n_{\sequent} = n_{\sequent_i}$.
\end{enumerate}
In particular, if $\infrule$ is a $\lr +$ step, we should choose the $\sequent_i$ admitting the smallest $n_{\sequent_i}$.
Now, since $P$ is a $\infty$-proof, $(\sequent_i)_{i<\omega}$ must be progressing and has infinitely many $\lr \mu$ steps.
We have two cases:
\begin{itemize}
    \item If $(\sequent_i)_{i<\omega}$ has infinitely many $\K$ steps then case \ref{item:k-step-strict-decrease} happens infinitely often along $(n_{\sequent_i})_{i<\omega}$, and so it is a monotone non-increasing sequent of natural numbers that does not converge. Contradiction.
    \item Otherwise $(\sequent_i)_{i\geq k}$ is $\K$-step-free, for some $k<\omega$, and so the number of letters in the LHS of the sequent is monotone non-decreasing in $i\geq k$.
    Since there are infinitely many $\lr \mu $ steps, by guardedness the number of producing expressions (whose languages necessarily are nonempty and do not contain $\epsilon$) is strictly increasing, and eventually dominates even $ n_0\geq n_i$. Contradiction. \qedhere
\end{itemize}
\end{proof}

By inspection of the rules of $\muHKA$ we have:
\begin{lemma}
    [Invertibility]
    \label{lem:invertibility}
    Let $\infrule$ be a logical step as in \eqref{eq:inf-step}.
    $\lang \Gamma \subseteq \lang S \implies \lang {\Gamma_i}\subseteq \lang {S_i}$, for each $i<k$.
\end{lemma}
\begin{theorem}
    [Completeness]
    \label{thm:completeness}
    $\lang \Gamma \subseteq \lang S \Rightarrow \muHKA \nwfproves \Gamma \seqar S$.
\end{theorem}
\begin{proof}[Proof sketch]
    We describe a bottom-up proof search strategy:
    \begin{enumerate}
        \item\label{item:apply-lhs-rules-maximally-mu-only} Apply left logical rules maximally, preserving validity by \cref{lem:invertibility}.
        Any infinite branch is necessarily progressing.
        \item\label{item:cfg-membership-prfs} This can only terminate at a sequent of the form $a_1, \dots, a_n \seqar S$ with $\vec a \in \lang S$, whence we mimic a `leftmost' parsing derivation for $\vec a$ wrt $S$. \qedhere
    \end{enumerate}
\end{proof}

To flesh out a bit more the \cref{item:cfg-membership-prfs} of the argument above, let us state:
\begin{lemma}
    [Membership]
    \label{lem:membership-mu-only}
    $a_1\cdots a_n \in \lang S \Rightarrow \muHKA \nwfproves a_1,\dots, a_n \seqar S$.
\end{lemma}
\begin{proof}
    [Proof idea]
    Recall the `canonical' grammar as constructed in \eqref{eq:can-grammar-mu-only}.
    We proceed by induction on a leftmost (or rightmost) derivation (as defined for Muller CFGs in \cref{sec:mem-sound-compl-gfps}) of $\vec a$ according to this canonical grammar.
\end{proof} %
\section{Completeness of an infinitary cut-free axiomatisation}
\label{sec:inf-ax-cf-theory}

While our completeness result above was relatively simple to establish, we can use it, along with proof theoretic techniques, to deduce completeness of an infinitary axiomatisation of the theory of $\mu$-expressions.
In fact we obtain an alternative proof of the result of \cite{GraHenKoz13:inf-ax-cf-langs}, strengthening it to a `cut-free' calculus $\muHKAw$. 

\medskip

Write $\CFA$ for the set of axioms consisting of:
\begin{itemize}
    \item $(0,1,+,\cdot)$ forms an idempotent semiring (aka a \emph{dioid}).
    \item \textbf{($\mu$-continuity)} $e\mu X f(X) g = \sum\limits_{n<\omega} ef^n(0)g$.
\end{itemize}
We are using the notation $f^n(0)$ defined by $f^0(0)\df 0$ and $f^{n+1}(0) \df f (f^n(0))$.
We also write $e\leq f$ for the natural order given by $e + f = f$. Now, define $\muHKAw$ to be the extension of $\muHKA$ by the `$\omega$-rule':
\[
\vlinf{\omega}{}{\Gamma , \mu X e(X), \Gamma'\seqar S}{\{ \Gamma, e^n(0),\Gamma'\seqar S \}_{n<\omega} }
\]
By inspection of the rules we have soundness of $\muHKAw $ for $\CFA$:
\begin{proposition}
\label{prop:soundness-of-muHKAw-for-CFA}
    $\muHKAw\proves \Gamma \seqar S \implies \CFA \proves \prod\Gamma \leq \sum\limits_{\Delta \in S}\prod \Delta$. 
\end{proposition}
Here the soundness of the $\omega$-rule above is immediate from $\mu$-continuity in $\CFA$.
Note, in particular, that $\CFA$ already proves that $\mu Xe(X)$ is indeed a fixed point of $e(\cdot)$, i.e.\ $e(\mu Xe(X)) = \mu X e(X)$ \cite{GraHenKoz13:inf-ax-cf-langs}.
The main result of this section is:
\begin{theorem}
\label{thm:muHKAnwf-to-muHKAw}
    $\muHKA\nwfproves e\seqar f$ $ \implies$ $ \muHKAw \proves e\leq  f$
\end{theorem}

Note that, immediately from \cref{thm:completeness} and \cref{prop:soundness-of-muHKAw-for-CFA}, we obtain:

\begin{corollary}
    $\lang e \subseteq \lang f \implies \muHKAw\proves e\leq f \implies \CFA \proves e\leq f$
\end{corollary}

To prove \cref{thm:muHKAnwf-to-muHKAw} we employ similar techniques to those used for an extension of \emph{linear logic} with least and greatest fixed points \cite{DDS23}, only specialised to the current setting.
We only sketch the ideas here, referencing the analogous definitions and theorems from that work at the appropriate moments.

\subsection{Projections}
Let us consider cedents $\Gamma = \Gamma(f_1, \dots, f_k)$ where some occurrences of $f_1 , \dots , f_k$ are distinguished. 
Note that the distinguished occurrences of each $f_i$ may include some, none or all of the occurrences of $f_i$ in $\Gamma$, including as subexpressions of of expressions in $\Gamma$.
We allow distinct $f_i$ and $f_j$ to be the same formula, as long as they distinguish a disjoint set of occurrences.

When  $\vec f = (\mu X_1 g_1(X_1), \dots, \mu X_k g(X_k))$ and $\vec n \in \Nat^k$, we write $\vec f^{\vec n} \df (g_1^{n_1}(0), \dots, g_k^{n_k}(0))$, the list obtained by \defname{assigning} $\vec n$ to $\vec f$.

Let us briefly recap the definition of \emph{projection}in~\cite[Definition 15]{DDS23}, specialised to our setting:

\begin{definition}
    [Projections]
    For each $\muHKA$ preproof $P$ of $\Gamma(\vec f) \seqar S$ and $\vec n \in \Nat^k$ we define $P(\vec n)$ a preproof of $\Gamma(\vec f^{\vec n}) \seqar S$ by coinduction on $P$:
    \begin{itemize}
        \item The definition of $P(\vec n)$ commutes with any step for which a distinguished formula is not principal for a $\lr\mu$ step.
        \item If $P$ ends with a step for which a distinguished formula if $\lr \mu $ principal,
        \begin{equation}
            \label{eq:mu-left-step-concludes-proof}
            \vlderivation{
        \vlin{\lr \mu}{}{\Gamma(\mu X e(X)), \mu X e(X), \Gamma'(\mu X e(X)) \seqar S}{
        \vltr{Q}{\Gamma(\mu X e(X)), e(\mu X e(X)), \Gamma'(\mu X e(X)) \seqar S}{\vlhy{\ }}{\vlhy{}}{\vlhy{\ }}
        }
        }
        \end{equation}
        then we proceed by case analysis on the number assigned:
        \[
        P(0,\vec n) \df \vlinf{\lr 0}{}{\Gamma(0), 0,\Gamma'(0) \seqar S}{}
        \]
        \[
        P(n+1,\vec n)\df 
            \toks0={.45}
            \vlderivation{
            \vlin{=}{}{\Gamma(e^{n+1}(0)), e^{n+1}(0), \Gamma'(e^{n+1}(0))\seqar S}{
            \vltrf{Q(n,n+1,\vec n)}{\Gamma(e^{n+1}(0)), e(e^{n}(0)),\Gamma'(e^{n+1}(0)) \seqar S}{\vlhy{\qquad\quad  }}{\vlhy{\qquad }}{\vlhy{\qquad \quad }}{\the\toks0}
            }
            } 
        \]
        where in the second case we have further distinguished the occurrences of $\mu X e(X)$ indicated in the auxiliary formula in \eqref{eq:mu-left-step-concludes-proof}.
    \end{itemize}
\end{definition}

Note that we have technically used a `repetition' rule $=$ in the translation above to ensure productivity of the translation.
However it turns out this is unnecessary, as $\infty$-proofs are indeed closed under taking projections:

\begin{proposition}
\label{prop:proj-pres-prog}
    If $P$ is a $\muHKA$ $\infty$-proof, then so is $P(\vec n)$.
\end{proposition}

The proof of this result follows the same argument as~\cite[Proposition 18]{DDS23}. To briefly recall the idea:
\begin{proof}
    [Proof sketch]
    Each maximal branch $B$ of $P(\vec n)$ is a prefix of some branch $B'$ of $P$, by inspection of the translation, only with some $\lr\mu$ steps replaced by `$=$' steps or, when $B$ is finite, a $\lr 0 $ step.
    Thus if $B$ is infinite then it must have infinitely many $\K$ steps, since $B'$ must be progressing, and so must have infinitely many $\lr \mu $ steps too.
\end{proof}

In particular we have:
\begin{corollary}
[Projection]
\label{lem:projection}
    For each $\infty$-proof $P$ of $ \Gamma, \mu X e(X),\Gamma' \seqar S$, $P(n)$ is an $\infty$-proof of $\Gamma, e^n(0), \Gamma' \seqar S$, for each $n < \omega$.
\end{corollary}

From here it is simple to provide a translation from $\muHKA$ $\infty$-proofs to $\muHKAw$ preproofs, as in \cref{def:omega-trans} shortly.
However, to prove the image of the translation is \emph{wellfounded}, we shall need some structural proof theoretic machinery, which will also serve later use when dealing with greatest fixed points in \cref{sec:gfps,sec:mem-sound-compl-gfps}.

\subsection{Intermezzo: ancestry and threads}
Given an inference step $\infrule$, as typeset in \cref{fig:muHKA-rules}, we
say a formula occurrence $f$ in an upper sequent is an \defname{immediate ancestor} of a formula occurrence $e$ in the lower sequent if they have the same colour; furthermore if $e$ and $f$ are occur in a cedent $\Gamma,\Gamma',\Delta,\Delta'$, they must be the matching occurrences of the same formula (i.e.\ at the same position in the cedent); similarly if $e$ and $f$ occur in the RHS context $S$, they must be matching occurrences in matching lists.

Construing immediate ancestry as a directed graph allows us to characterise progress by consideration of its paths:
\begin{definition}
    [(Progressing) threads]
    \label{def:threads}
    Fix a preproof $P$. 
    A \defname{thread} is a maximal path in the graph of immediate ancestry. 
    An infinite thread on the LHS is \defname{progressing} if it is infinitely often principal (i.o.p.) for a $\lr \mu $ step.
\end{definition}

Our overloading of terminology is suggestive:
\begin{proposition}
\label{prop:lmu-fair-implies-exists-prog-thread}
    $P$ is progressing $\Leftrightarrow$ each branch of $P$ has a progressing thread.
\end{proposition}
\begin{proof}
    [Proof sketch]
    The $\impliedby $ direction is trivial.
    For $\implies$ direction we appeal to K\"onig's lemma.
    Fix a branch $B$ and take the subtree of its immediate ancestry graph with nodes the principal formula occurrences along $B$, and edges given by reachability (`direct ancestry'). 
    By progressiveness this tree is infinite, and by inspection of the rules it is finitely branching, thus it must have an infinite path by K\"onig's Lemma. 
    This path induces an infintely often principal thread along $B$, which in turn must be infinitely often $\lr\mu$ principal as every other LHS rule strictly decreases the size of formula.
\end{proof}

\begin{example}
    Recall the $\infty$-proof in \cref{ex:anbn-in-asbs}. 
    The only infinite branch, looping on $\bullet$, has a progressing thread indicated in \mgt{magenta}.
\end{example}

\begin{fact}
[See, e.g., \cite{Koz83:results-on-mu,KupMarVen22:graph-reps-mu-forms}]
\label{fact:io-threads-have-critical-formula}
    Any i.o.p.\ thread has a unique smallest i.o.p.\ formula, under the subformula relation.
    This formula must be a fixed point formula.
\end{fact}

\subsection{Translation to \texorpdfstring{$\omega$}{}-branching system}

We are now ready to give a translation from $\muHKAnwf$ to $\muHKAw$.

\begin{definition}
    [$\omega$-translation]
    \label{def:omega-trans}
    For preproofs $P$ define $P^\omega$ by coinduction:
    \begin{itemize}
        \item $\cdot^\omega$ commutes with any step not a $\lr\mu$.
        {\smallskip}
        \item $\left( 
        \vlderivation{
        \vlin{\lr\mu}{}{\Gamma, \mu X e(X) , \Gamma' \seqar S}{
        \vltr{P}{\Gamma, e(\mu Xe(X)),\Gamma' \seqar S}{\vlhy{\quad }}{\vlhy{}}{\vlhy{\quad }}
        }
        }
        \right)^\omega 
        \ \df \ 
        \vlderivation{
        \vlin{\omega}{}{\Gamma, \mu X e(X),\Gamma' \seqar S}{
        \vlhy{
        \left\{
        \vltreeder{P(n)^\omega}{\Gamma, e^n(0),\Gamma' \seqar S}{\quad}{\ }{\quad }
        \right\}_{n<\omega}
        }
        }
        }
        $
    \end{itemize}
\end{definition}

 \cref{thm:muHKAnwf-to-muHKAw} now follows immediately from the following result, obtained by analysis of progressing threads in the image of the $\omega$-translation:
\begin{lemma}
\label{lem:omega-trans-is-wf}
     $P$ is progressing $\implies$ $P^\omega$ is wellfounded.
\end{lemma}

The proof of~\Cref{lem:omega-trans-is-wf} follows the same argument as for the analogous result in~\cite[Lemma 23]{DDS23}.

\begin{proof}
    [Proof sketch]
    Each branch of $P^\omega$ can be specified by a branch $B$ of $P$ and a (possibly infinite) sequence of natural numbers $\vec n \in \omega^{\leq \omega}$, specifying at each $\omega$-step which premiss to follow. Call this $B^{\vec n}$.
    Now, take a progressing thread along $B$ and consider its smallest i.o.p.\ $\mu$-formula, say $\mu X e(X)$.
    If $B^{\vec n}$ follows, say the $k$\textsuperscript{th} premiss at the first corresponding principal step for $\mu X e(X)$ in $B$, write $K$ for height of the $(k+1)$\textsuperscript{th} unfolding of $\mu X e(X)$ in $B$.
    It follows from inspection of the $\omega$-translation and projections that $B^{\vec n}$ has height $\leq K$.
\end{proof}

\begin{example}
    Recalling \cref{ex:anbn-in-asbs}, let us see the $\omega$-translation of $R$ in \eqref{eq:anbn-in-asbs}.
    First, let us (suggestively) write $\anbnapprx n $ for the $n$\textsuperscript{th} approximant of $\anbn$, i.e.\ $\anbnapprx 0 \df 0$ and $\anbnapprx {n+1} \df 1 + a \anbnapprx n b$. 
    Now $R^\omega$ is given below, left, where recursively $R(0) \df \vlinf{\lr 0}{}{0\seqar \asbs}{}$ and $R(n+1)$ is given below, right:
    \[
    \newcommand{\storage}{R(n)}
    \vlinf{\omega,\rr\mu}{}{\mgt\anbn \seqar [\asbs]}{
    \left\{
    \vltreeder{\storage}{\mgt{\anbnapprx n} \seqar [\asbs]}{\quad}{}{\quad }
    \right\}_{n<\omega}
    }
    \quad
    ;
    \ 
        \vlderivation{
        \vliin{\lr + }{}{\mgt{1 + a \anbnapprx n b} \seqar [\asbs]}{      
            \vliq{}{}{1  \seqar [\asbs ]}{
            \vlin{\lr 1 , \rr 1}{}{1 \seqar [1]}{
            \vlin{\id}{}{\seqar [\, ]}{\vlhy{}}
            }
            }
        }{
            \vliq{}{}{\mgt{a\anbnapprx n b} \seqar [ \asbs ]}{
            \vlin{\lr \cdot, \rr \cdot}{}{\mgt{a\anbnapprx n  b} \seqar [  a\asbs  ]}{
            \vlin{\kl a}{}{a,\mgt{\anbnapprx n  b} \seqar [  a,\asbs  ]}{
            \vlin{\rr \mu,\rr +,\wk}{}{\mgt{\anbnapprx n  b} \seqar [  \asbs  ]}{
            \vlin{\lr \cdot, \rr \cdot}{}{\mgt{\anbnapprx n b} \seqar [  \asbs b  ]}{
            \vlin{\kr b}{}{\mgt{\anbnapprx n } ,b \seqar [  \asbs ,b  ]}{
            \vltr{\storage}{\mgt{\anbnapprx n } \seqar \asbs}{\vlhy{\quad}}{\vlhy{}}{\vlhy{\quad}}
            }
            }
            }
            }
            }
            }
}
        }
    \]
\end{example} %
\section{Greatest fixed points and \texorpdfstring{$\omega$}{}-languages}
\label{sec:gfps}

We extend the grammar of expressions from \eqref{eq:grammar-of-mu-exprs} by:
\[
e,f \dots \quad \bnf \quad
\dots \quad \mid \quad \nu X e(X)
\]
We call such expressions \emph{$\mu\nu$-expressions} when we need to distinguish them from ones without $\nu$.
The notions of a \emph{(left-)productive} and \emph{(left-)guarded} expression are defined in the same way, extending the grammar of \eqref{eq:prod-exprs-mu-only} by the clause $\nu X p$.

As expected $\mu\nu$-expressions denote languages of finite and infinite words:
\begin{definition}
    [Intended semantics of $\mu\nu$-expressions]
We extend the notation $vw$ to all $v,w\in \Alphabet^{\leq \omega}$ by setting $vw = v$ when $|v|= \omega$.
    We extend the definition of $\lang \cdot $ from \cref{def:reg-lang-model} to all $\mu\nu$-expressions by setting $\wlang {\nu X e(X)} \df \bigcup \{ A \subseteq \wlang {e(A)} \}$ where now $A$ varies over subsets of $\Alphabet^{\leq \omega}$.
\end{definition}
Again, since all the operations are monotone, $\wlang{\nu X e(X)}$ is indeed the greatest fixed point of the operation $A\mapsto \lang{e(A)}$, by the Knaster-Tarski theorem.
In fact ($\omega$-)languages computed by $\mu\nu$-expressions are just the `$\omega$-context-free languages' ($\omega$-CFLs), cf.~\cite{CG77,Linna76}, defined as the {`Kleene closure'} of CFLs:

\begin{definition}[$\omega$-context-free languages]
For $A\subseteq \Alphabet^{+}$ we write $A^\omega \df \{w_0w_1w_2\cdots : \forall i<\omega \, w_i \in A\}$.
The class of \defname{$\omega$-CFLs} ($\oCF$) is defined by:
\[\oCF\ :=\ \left\{\bigcup_{i<n} A_iB_i^{\omega} \ : \  n<\omega ; \, A_i,B_i \text{ context-free and } \epsilon \notin A_i,B_i, \, \forall i<n \right\}\]    
\end{definition}

It is not hard to see that each $\omega$-CFL is computed by a $\mu\nu$-expression, by noting that $\lang e^\omega = \lang {\nu X (eX)}$:

\begin{proposition}
    \label{prop:ocf-is-munu-def}
    $L\in \oCF \implies L = \lang e$ for some left-productive $e$.
\end{proposition}
\begin{proof}
Given $L=\bigcup_{i=1}^n A_iB_i^{\omega}$, by~\cref{prop:cf-are-mu-definable} let $e_i$ and $f_i$ be left-productive with $\lang{e_i}=A_i$, $\lang{f_i}=B_i$ for $i=1,\dots, n$.
Then, $L=\lang{\sum_{i=1}^n e_i\nu X(f_iX)}$.    
\end{proof}

We shall address the converse of this result later.
First let us present our system for $\mu\nu$-expressions, a natural extension of $\muHKA$ earlier:

\medskip

\begin{definition}
    [System]
    The system $\munuHKA$ extends $\muHKA$ by the rules:
\begin{equation}
    \label{eq:nu-rules}
    \vlinf{\lr\nu}{}{\orange \Gamma, \mgt{\nu X e(X)}, \olv{ \Gamma'} \seqar \purple S}{\orange \Gamma, \mgt{e(\nu X e(X))}, \olv {\Gamma'} \seqar \purple S}
\quad
\vlinf{\rr\nu}{}{\orange\Gamma \seqar \purple S, [\blue \Delta, \mgt{\nu X e(X)}, \green {\Delta'}]}{\orange \Gamma \seqar \purple S, [\blue \Delta, \mgt{e(\nu X e(X)}, \green {\Delta'}]}
\end{equation}
\emph{Preproofs} for this system are defined just as for $\muHKA$ before.
The definitions of \emph{immediate ancestor} and \emph{thread} for $\munuHKA$ extends that of $\muHKA$ from \cref{def:threads} according to the colouring above in \eqref{eq:nu-rules}.
\end{definition}

\noindent
However we must be more nuanced in defining progress, requiring a definition at the level of threads as in \cref{sec:inf-ax-cf-theory}.
Noting that Fact~\ref{fact:io-threads-have-critical-formula} holds for our extended language with $\nu$s as well as $\mu$s, we call an i.o.p.\  thread a \defname{$\mu$-thread} (or \defname{$\nu$-thread}) if its smallest i.o.p.\ formula is a $\mu$-formula (or $\nu$-formula, respectively).

\begin{definition}
    [Progress]
    Fix a preproof $P$. 
    We say that an infinite thread $\tau$ along a (infinite) branch $B$ of $P$ is \defname{progressing} if it is i.o.p.\ and it is a $\mu$-thread on the LHS or it is a $\nu$-thread on the RHS.
    $B$ is \defname{progressing} if it has a progressing thread.
    $P$ is a $\infty$-\defname{proof} of $\munuHKA$ if each of its infinite branches has a progressing thread.
\end{definition}

\begin{figure}
    \[
\vlderivation{
\vlin{\rr \mu}{}{e \seqar [\mgt f]}{
\vlin{\rr +}{}{e \seqar [\mgt{b+ \nu X(afX)}]}{
\vlin{\wk}{}{e \seqar [b],[\mgt{\nu X(afX)}]}{
\vlin{\rr \nu}{\bullet}{e \seqar [\underline{\mgt{\nu X(afX)}}]}{
\vliq{\lr \cdot, \rr \cdot}{}{abe \seqar [\mgt{af\nu X (afX)}]}{
\vlin{\kk a}{}{a,b,e \seqar [a,f,\mgt{\nu X (a f X)}]}{
\vlin{\rr \mu}{}{b,e \seqar [f,\mgt{\nu X (a f X)}]}{
\vlin{\rr +}{}{b,e \seqar[b+ \nu X(aYX),\mgt{\nu X(afX)}] }{
\vlin{\wk }{}{b,e \seqar [b,\mgt{\nu X (afX)}],[\nu X(aYX),\nu X(afX)] }{
\vlin{\kk b}{}{b,e\seqar [b,\mgt{\nu X (a f X)}] }{
\vlin{\rr \nu}{\bullet}{e \seqar [\mgt{\nu X (afX)}] }{\vlhy{\vdots}}
}
}
}
}
}
}
}
}
}
}
}
\]
    \caption{A $\lmunuHKA$ $\infty$-preproof of $e\seqar [f]$, where $e \df \nu Z (abZ) $ and $f \df \mu Y(b + \nu X (aYX))$.}
    \label{fig:braiding-example}
\end{figure}

\begin{example}
\label{ex:r1-example}
Write $e \df \nu Z (abZ) $ and $f \df \mu Y(b + \nu X (aYX))$.
The sequent $e\seqar [f]$ has a preproof given in \cref{fig:braiding-example}.
This preproof has just one infinite branch, looping on $\bullet$, which indeed has a progressing thread following the \mgt{magenta} formulas. 
The only fixed point infinitely often principal along this thread is $\nu X (afX)$, which is principal at each $\bullet $. Thus this preproof is a proof and $e\seqar [f]$ is a theorem of $\lmuHKAnwf$.

Note that, even though this preproof is progressing, the infinite branch's smallest i.o.p.\ formula on the RHS is \emph{not} a $\nu$-formula, e.g.\ given by the \mgt{magenta} thread, as $f$ is also i.o.p. Let us point out that (a) the progressiveness condition only requires \emph{existence} of a progressing thread, even if other threads are not progressing (like the unique LHS thread above).
\end{example}

\subsection*{Some necessary conventions: left-guarded and leftmost}
Crucially, due to the asymmetry in the definition of the product of infinite words, we must employ further conventions to ensure soundness and completeness of $\infty$-proofs for $\lang \cdot$.
Our choice of conventions is inspired by the usual `leftmost' semantics of `$\omega$-CFGs', which we shall see in the next section.

\smallskip

First, we shall henceforth work with a \emph{lefmost} restriction of $\munuHKA$ in order to maintain soundness for $\lang \cdot$:
\begin{definition}
    A $\munuHKA$ preproof is \defname{leftmost}
    if each logical step has principal formula the leftmost formula of its cedent, and there are no $\kr{}$-steps.
    Write $\lmunuHKA$ for the restriction of $\munuHKA$ to only leftmost steps and $\lmunuHKAnwf$ for the class of $\infty$-proofs of $\lmunuHKA$.
\end{definition}

\noindent
 We must also restrict ourselves to left-guarded expressions in the sequel:

\begin{convention}
Henceforth, all expressions are assumed to be left-guarded.
\end{convention}

\noindent
Let us justify both of these restrictions via some examples.

\begin{remark}
    [Unsound for non-leftmost]
    Unlike the $\mu$-only setting it turns out that $\munuHKAnwf$ is unsound without the leftmost restriction, regardless of left-guardedness. 
    For instance consider the preproof,
    \[
    \vlderivation{
    \vlin{\rr\nu}{\bullet}{\seqar [\mgt{\nu X (aX)}]}{
    \vlin{\rr\cdot}{}{\seqar [\mgt{a\nu X(aX)}]}{
    \vlin{\rr\nu}{a,\bullet}{\seqar [a,\mgt{\nu X (aX)}]}{\vlhy{\vdots}}
    }
    }
    }
    \]
    where $a,\bullet$ roots the same subproof as $\bullet$, but for an extra $a$ on the left of every RHS.
    Of course the endsequent is not valid, as the LHS denotes $\{\epsilon\}$ while the RHS denotes $\{a^\omega\}$.
    Note also that, while it is progressing thanks to the thread in \mgt{magenta}, it is not leftmost due to the topmost displayed $\rr\nu$ step.
\end{remark}

\begin{remark}
    [Incomplete for unguarded]
    On the other hand, without the left-guardedness restriction, $\lmunuHKAnwf$ is not complete. 
    For instance the sequent $\nu X X \seqar [\;] , \{ [a , \nu X X]\}_{a \in \Alphabet}$ is indeed valid as both sides compute all of $\pow{\Alphabet^{\leq \omega}}$: any word is either empty or begins with a letter.
    However the only available (leftmost) rule application, bottom-up, is $\lr \nu$, which is a fixed point of leftmost proof search, obviously not yielding a progressing preproof.
\end{remark}

\section{Metalogical results: a game-theoretic approach}
\label{sec:mem-sound-compl-gfps}

Now we return to addressing the expressiveness of both the syntax of $\mu\nu$-expressions and our system $\lmunuHKAnwf$, employing game-theoretic methods.

\subsection{Evaluation puzzle and soundness}
\begin{figure}[t]
    \centering
    \begin{tabular}{|c|c|}
        \hline
        Position & Available move(s)
        \\\hline
        $(aw,[a,\Delta])$ & $(w,\Delta)$\\
        $(w,[1,\Delta])$ & $(w,\Delta)$\\
        $(w,[e+f,\Delta])$ & $(w,[e,\Delta]),(w,[f,\Delta])$ \\
        $(w,[ef,\Delta])$ & $(w,[e,f,\Delta])$ \\
        $(w,[\mu Xf(X),\Delta])$ & $(w,[f(\mu Xf(X)),\Delta])$ \\
        $(w,[\nu Xf(X),\Delta])$ & $(w,[f(\nu Xf(X)),\Delta])$ \\
        \hline
    \end{tabular}    
    \caption{Rules of the evaluation puzzle.}
    \label{fig:eval-rules}
\end{figure}

As an engine for our main metalogical results about $\lmunuHKA$, and for a converse to \cref{prop:ocf-is-munu-def}, we first characterise membership via games:

\begin{definition}
The \defname{evaluation puzzle} is a puzzle (i.e.\ one-player game) whose positions are pairs $(w,\Gamma)$ where $w\in \Alphabet^{\leq \omega}$ and $\Gamma$ is a cedent, i.e.\ a list of $\mu\nu$-expressions. A \defname{play} of the puzzle runs according to the rules in \cref{fig:eval-rules}: puzzle-play is deterministic at each state except when the expression is a sum, in which case a choice must be made. During a play of the evaluation puzzle, formula ancestry and threads are defined as for $\lmunuHKA$ preproofs,  by associating each move with the LHS of a left logical rule. 
A play is \defname{winning} if:
\begin{itemize}
    \item it terminates at the winning state $(\epsilon,[\;])$; or,
    \item it is infinite and has a $\nu$-thread (along its right components).
\end{itemize}
\end{definition}

\begin{example}
Define $d\df \mu X(\langle\rangle + \langle X\rangle X)$, the set of non-empty well-bracketed words. Let $d^\omega\df \nu Y dY$. Let us look at a play from $(\langle^\omega,[d^\omega])$. 

\begin{center}
\begin{tikzcd}[cramped, sep=small] 
\arrow[r] & \scriptstyle(\langle^\omega,[d^\omega])
\arrow[r] & \scriptstyle(\langle^\omega,[dd^\omega])
\arrow[r] & \scriptstyle(\langle^\omega,[d,d^\omega])
\arrow[r] & \scriptstyle(\langle^\omega,[\langle\rangle+\langle d\rangle d,d^\omega]) 
\arrow[r] & \scriptstyle(\langle^\omega,[\langle d\rangle d,d^\omega]) \arrow[d]\\
& & \dots & \scriptstyle(\langle^\omega,[d,\rangle d,d^\omega]) \arrow[l] & \scriptstyle(\langle^\omega,[d\rangle d,d^\omega]) \arrow[l] & \scriptstyle(\langle^\omega,[\langle, d\rangle d,d^\omega]) \arrow[l] 
\end{tikzcd}
\end{center}

The play continues without $d^\omega$ ever being principal (essentially, going into deeper and deeper nesting to match a $\langle$ with a $\rangle$). Since even the first match is never made, there is no hope of progress. The play (and, in fact, any play) is thus losing. On the other hand, consider the following play from $(u,[d^\omega])$ where $u=(\langle\rangle)^\omega$. It is, indeed winning and its $\nu$-thread is indicated in \mgt{magenta}.

\begin{center}
\begin{tikzcd}[cramped, sep=small] 
\arrow[r] & \scriptstyle(u,[\mgt{d^\omega}])
\arrow[r,"_2" near end] & \scriptstyle(u,[d,\mgt{d^\omega}])
\arrow[r] & \scriptstyle(u,[\langle\rangle+\langle d\rangle d,\mgt{d^\omega}]) 
\arrow[r] & \scriptstyle(u,[\langle\rangle,\mgt{d^\omega}]) 
\arrow[r] & \scriptstyle(u,[\langle,\rangle,\mgt{d^\omega}]) 
\arrow[r] &\scriptstyle(\rangle u,[\rangle,\mgt{d^\omega}]) \arrow[lllll,in=-10, out=-170]
\end{tikzcd}
\end{center}

\end{example}

\begin{theorem}
[Evaluation]
\label{thm:evaluation-theorem}
    $w \in \wlang\Gamma$ $\Leftrightarrow $ there is a winning play from $(w,\Gamma)$.
\end{theorem}

The proof is rather involved, employing the method of `signatures' common in fixed point logics, cf.\ e.g.\ \cite{NiwWal96:games-mu-calc}, which serve as `least witnesses' to word membership via carefully managing \emph{ordinal approximants} for fixed points. 
Here we must be somewhat more careful in the argument because positions of our puzzle include \emph{cedents}, not single formulas: we must crucially assign signatures to \emph{each} formula of a cedent.
Working with cedents rather than formulas allows the evaluation puzzle to remain strictly single player. This is critical for expressivity: \emph{alternating} context-free grammars and pushdown automata compute more than just CFLs~\cite{MHHO05,ChaKozSto81:alternation}.

The next subsection is devoted to a proof of the Evaluation Theorem above.
Before that, let us give an important consquence:

\smallskip
We can now prove the soundness of $\lmunuHKAnwf$ by reduction to \cref{thm:evaluation-theorem}:

\begin{theorem}[Soundness]
\label{thm:munuHKA-soundness}
$\lmunuHKA\nwfproves \Gamma\seqar S$ $\implies$ $\lang{\Gamma}\subseteq\lang{S}$. 
\end{theorem}

\begin{proof}[Proof sketch]
Let $P$ be a $\infty$-proof of $\Gamma\seqar S$ and $w\in \wlang\Gamma$. We show $w\in\wlang S$. 
First, since $w\in \wlang\Gamma$ there is a winning play $\pi$ from $(w,\Gamma)$ by \cref{thm:evaluation-theorem}, which induces a unique (maximal) branch $B_\pi$ of $P$ which must have a progressing thread $\tau$.
Now, since $\pi$ is a \emph{winning} play from $(w,e)$, $\tau$ cannot be on the LHS, so it is an RHS $\nu$-thread following, say, a sequence of cedents $[\Gamma_i]_{i<\omega}$.
By construction $[\Gamma_i]_{i<\omega}$ has an infinite subsequence, namely whenever it is principal, that forms (the right components of) a winning play from $(w,\Gamma_0)$, with $\Gamma_0 \in S$. Thus indeed $w \in \lang S$ by \cref{thm:evaluation-theorem}.
\end{proof}

\subsection{Proof of the Evaluation Theorem}
Let us now revisit the argument for \cref{thm:evaluation-theorem} more formally.
This subsection may be safely skipped by the reader comfortable with that result.

Write $\subform$ for the subformula relation.
Recalling \cref{def:fl}, let us write $e \leqfl f$ if $e \in \fl f$, 
$e \lefl f$ if $e\leqfl f \not\leqfl e$ and $e  \eqfl f$ if $e\leqfl f \leqfl e$.

\begin{definition}[Dependency order]
Let the \defname{dependency order} be $\dleq\, \df\, \leqfl \times \supform$, i.e.\ $e\dleq f$ if either $e\lefl f$ or $e\eqfl f $ and $f \subform e$.
\end{definition}

Note that, by the properties of FL closure, $\dleq$ is a well partial order on expressions. In the sequel, we assume an arbitrary extension of $\dleq$ to a total well-order $\leq$.

\begin{definition}[Signatures]
Let $M$ and $N$ be finite sets of $\mu$-expressions $\{\mu X_0 e_0  \dge \cdots \dge \mu X_{n-1} e_{n-1}\}$ and $\{\nu X_0 e_0  \dge \cdots \dge \nu X_{n-1} e_{n-1}\}$ respectively. A $\mu$-\defname{signature} (respectively, a $\mu$-\defname{signature}) is a sequence $\vec \alpha = (\alpha_i)_{i=0}^{n-1}$ of ordinals indexed by $M$ (respecively, $N$). Signatures are ordered by the lexicographical product order.
\end{definition}

Let us temporarily expand the language of expressions by,
\[
e,f, \dots \quad \bnf \quad
\dots \quad \mid \quad 
\mu^\alpha X e \quad \mid \quad
\nu^\alpha X e
\]
where $\alpha$ ranges over ordinals.

\begin{definition}
Fix a finite set of $\mu$-formulas $\{\mu X_0 e_0  \dge \cdots \dge \mu X_{n-1} e_{n-1}\}$ and a $\mu$-signature $\vec\alpha=(\alpha_i)_{i=0}^{n-1}$. Given an expression $e$, its corresponding $\mu$-\defname{signed formula} $e^{\vec\alpha}$ is one where every occurrence $\mu X_i e_i$ has been replaced by $\mu^{\alpha_i} X_i e_i$.
\end{definition}

Similarly, given a finite set of $\nu$-formulas, a $\nu$-signature $\vec\alpha$, and an expression $e$, its corresponding $\nu$-\defname{signed formula} $e_{\vec\alpha}$ is one where every occurrence of a $\nu$-subformula has been replaced by its corresponding approximant. 

We interpret such expressions by the inflationary and deflationary constructions respectively:
\begin{itemize}
    \item $\wlang {\mu^0 X e(X) } \df \emptyset$
    \item $\wlang {\mu^{\alpha +1} X e(X)} \df \wlang {e(\mu^\alpha X e(X))}$
    \item $\wlang {\mu^\lambda X e(X)} \df \bigcup\limits_{\alpha <\lambda} \wlang {\mu^\alpha X e(X)}$, when $\lambda $ is a limit ordinal. 
    \medskip
    \item $\wlang {\nu_0 X e(X) } \df \Alphabet^{\leq \omega}$
    \item $\wlang {\nu_{\alpha +1} X e(X)} \df \wlang {e(\nu^\alpha X e(X))}$
    \item $\wlang {\nu_\lambda X e(X)} \df \bigcap\limits_{\alpha <\lambda} \wlang {\nu^\alpha X e(X)}$, when $\lambda $ is a limit ordinal. 
\end{itemize}

Finally, we extend the notion of $\mu$ and $\nu$ signatures to lists of expressions by writing $\Gamma^{\lv\alpha}$ and $\Gamma_{\lv\alpha}$ (parameterised by \emph{lists} of vectors of ordinals now, by abuse of notation) for the corresponding $\mu$ and $\nu$-signed lists of $\Gamma$. Spelt out, $[e_1,\dots,e_n]^{\vec\alpha_1::\dots::\vec\alpha_n}$ is a shorthand for $[e_1^{\vec\alpha_1},\dots,e_n^{\vec\alpha_n}]$ (similarly for $\nu$-signed). Lists of vectors are lexicographically ordered.

Recall that least and greatest fixed points can be computed as limits of approximants. 
In particular, we have,
\begin{itemize}
    \item $\lang{\mu Xe} = \bigcup \limits_{\alpha\in\SF{Ord}} \lang{\mu^{\alpha} Xe}$
    \item $\lang{\nu Xe} = \bigcap \limits_{\alpha\in\SF{Ord}}\lang{\nu^{\alpha} Xe}$
\end{itemize}
where $\alpha $ ranges over ordinals. Thus we have immediately:

\begin{proposition}
\label{prop:least-signature}
Suppose $\Gamma$ is a list of expressions. We have:
\begin{itemize}
    \item If $w \in \wlang\Gamma$ then there is a $\mu$-signature $\lv\alpha$ such that $w \in \wlang {\Gamma^{\lv\alpha}}$.
    \item If $w\notin \wlang\Gamma$ then there is a $\nu$-signature $\lv\alpha$ such that $w \notin \wlang {\Gamma_{\lv\alpha}}$.
\end{itemize}
\end{proposition}

We are now ready to prove our characterisation of evaluation:

\begin{proof}
[Proof of \cref{thm:evaluation-theorem}]

    $(\implies)$ 
    Suppose $w \in \wlang\Gamma$. By~\Cref{prop:least-signature}, there is a least $\mu$-signature $\lv \alpha$ such that $w \in \wlang {\Gamma^{\lv \alpha}}$. We will construct a winning play $(w_i,\Gamma_i)_{i\in\lambda}$ and sequence of signatures $(\lv{\alpha^i})_{i\in\lambda}$ from $(w,\Gamma)$ such that:
    \begin{itemize}
        \item $(w_0,\Gamma_0,\vec{\alpha_0})=(w,\Gamma,\lv\alpha)$;
        \item for all $i\in\lambda$, $\lv{\alpha^i}$ is a $\mu$-signature such that $w_i \in \wlang {\Gamma^{\lv{\alpha^i}}}$;
    \end{itemize}
where $\lambda\in\omega+1$. If the play is finite then it is winning by construction, so assume it is infinite i.e. $\lambda=\omega$.

We will construct it by induction on $i$. The base case is already defined. For the induction case, assume $\Gamma_i=f,\Delta$ and $\lv\alpha_i=\vec{\alpha}::\lv\alpha$ and we will do a case -analysis on $f$.

\begin{itemize}
    \item Suppose $f=f_0+f_1$.

    \begin{align*}
        &w_i \in\lang{[f_0+f_1,\Delta]^{\vec{\alpha}::\lv\alpha}} &&[\text{By induction hypothesis}]\\
        &\implies w_i\in\lang{[f_0+f_1^{\vec{\alpha}},\Delta^{\lv\alpha}]}\\
        &\implies w_i\in \lang{[f_0^{\vec{\alpha}},\Delta^{\lv\alpha}]} \text { or }w_i\in\lang{[f_1^{\vec{\alpha}},\Delta^{\lv\alpha}]}\\
        &\implies w_i\in\lang{[f_0,\Delta]^{\lv{\alpha^i}}} \text { or }w_i\in\lang{[f_1,\Delta]^{\lv{\alpha^i}}}
    \end{align*}

    Wlog, assume $w_i\in\lang{[f_0,\Delta]^{\lv{\alpha^i}}}$. Choose $(w_{i+1},\Gamma_{i+1},\lv{\alpha^{i+1}})=(w_i,[f_0,\Delta],\lv{\alpha_i})$.

\item Suppose $f=f_0\cdot f_1$.

    \begin{align*}
        &w_i \in\lang{[f_0\cdot f_1,\Delta]^{\vec{\alpha}::\lv\alpha}} &&[\text{By induction hypothesis}]\\
        &\implies w_i\in\lang{[f_0^{\vec{\alpha}},f_1^{\vec{\alpha}},\Delta^{\lv\alpha}]}\\
        &\implies w_i\in \lang{[f_0,f_1,\Delta]^{\vec{\alpha}::\vec{\alpha}::\lv{\alpha}}}
    \end{align*}

 Choose $(w_{i+1},\Gamma_{i+1},\lv{\alpha^{i+1}})=(w_i,[f_0,f_1,\Delta],\vec{\alpha}::\vec{\alpha}::\lv{\alpha})$.

\item Suppose $f=a$ for some $a\in\Alphabet$. %

    \begin{align*}
        &w_i \in\lang{[a,\Delta]^{\vec\alpha::\lv\alpha}} &&[\text{By induction hypothesis}]\\
        &\implies w_i\in\lang{[a^{\vec\alpha},\Delta^{\lv\alpha}]}\\
        &\implies w_i=aw' \text{ and } w'\in\lang{\Delta^{\lv\alpha}}
    \end{align*}

 Choose $(w_{i+1},\Gamma_{i+1},\lv{\alpha^{i+1}})=(w',\Delta,\lv{\alpha})$.

\item Suppose $f=\nu X f_0$.

    \begin{align*}
        &w_i \in\lang{[\nu Xf_0,\Delta]^{\vec\alpha::\lv\alpha}} &&[\text{By induction hypothesis}]\\
        &\implies w_i\in\lang{[\nu Xf_0^{\vec{\alpha}},\Delta^{\lv{\alpha}}]}\\
        &\implies w_i\in\lang{[f_0(\nu Xf_0)^{\vec{\alpha}},\Delta^{\lv{\alpha}}]} && [\lang{\nu X e, \Gamma}=\lang{e(\nu X e), \Gamma}]\\
        &\implies w_i\in \lang{[f_0(\nu Xf_0),\Delta]^{\vec{\alpha}::\lv\alpha}}
    \end{align*}

Choose $(w_{i+1},\Gamma_{i+1},\lv{\alpha^{i+1}})=(w_i,[f_0(\nu Xf_0),\Delta],\lv{\alpha^{i}})$

\item Suppose $f=\mu X f_0$.

    \begin{align*}
        &w_i \in\lang{[\mu Xf_0,\Delta]^{\vec{\alpha}::\lv\alpha}} &&[\text{By induction hypothesis}]\\
        &\implies w_i\in\lang{[\mu^{\alpha_j} Xf_0,\Delta^{\lv{\alpha}}]} && [\text{where }\vec\alpha=(\alpha_i)]\\
        &\implies w_i\in\lang{[f_0(\mu^{\beta} Xf_0),\Delta^{\lv{\alpha}}]} && [\text{for some }\beta<\alpha_j]\\
        &\implies w_i\in\lang{[f_0(\mu Xf_0),\Delta]^{\vec\alpha'::\lv{\alpha}}} && [\vec\alpha'=\vec\alpha[\beta/\alpha_j]]
    \end{align*}

Choose $(w_{i+1},\Gamma_{i+1},\lv{\alpha^{i+1}})=(w_i,[f_0(\nu Xf_0),\Delta],\vec\alpha'::\lv{\alpha})$ works.
\end{itemize}
We now claim that this play is winning. Suppose not. Then, the smallest formula principal infinitely often is $\mu$. We follow its thread and obtain a strictly decreasing sequence of ordinals. Contradiction.

\medskip

$(\impliedby)$ For the converse direction, we will prove the contrapositive. Suppose $w\not\in\wlang\Gamma$. Fix an arbitrary play $\pi=(w_i,\Gamma_i)_{i\in\lambda}$ from $(w,\Gamma)$ for some $\lambda\in\omega+1$. By inspection of the puzzle rules, non-membership is preserved i.e.\ we always have $w_i \notin \wlang {\Gamma_i}$. By~\cref{prop:least-signature}, there are $\nu$-signatures $\lv{\alpha_i}$ such that $w_i \notin \wlang {\Gamma_{\lv{\alpha_i}}}$. Following an argument like above, the signature corresponding to a thread is a monotone non-increasing sequence. Moreover, if $\pi$ is winning, it is strictly decreasing. Therefore, $\pi$ is not winning.
\end{proof}

\subsection{\texorpdfstring{$\omega$}{}-context-freeness via Muller grammars}
We can now use the adequacy of the evaluation puzzle to recover a converse of \cref{prop:ocf-is-munu-def}. For this, we need to recall a grammar-formulation of $\oCF$, due to Cohen and Gold \cite{CG77} and independently Nivat \cite{Nivat1977,Nivat1978}.

\smallskip

A \defname{Muller ($\omega$-)CFG} (MCFG) is a CFG $\gram$, equipped with a set $F\subseteq \pow{\Var}$ of \defname{acceptable} sets. 
We define a rewrite relation $\lred \gram \, \subseteq (\Var \cup \Alphabet)^* \times (\Var \cup \Alphabet)^*$, \defname{leftmost reduction}, by
\(
\vec a X v \lred\gram \vec a u v 
\)
whenever $\vec a \in \Alphabet^*$, $X\to u$ is a production of $\Cal G$ and $v \in (\Var \cup \Alphabet)^*$.
A \defname{leftmost derivation} is just a maximal (possibly infinite) sequence along $\lred\gram $.
We say
$\gram $ \defname{accepts} $w \in \Alphabet^{ \leq \omega}$ if there is a leftmost derivation $\delta$ such that $\delta$ converges to $w$ and the set of infinitely often occurring states that are LHSs of productions along $\delta$ is in $F$.
We write $\lang {\gram } $ for the set of words $\gram $ accepts.

\begin{theorem}[\cite{CG77,Nivat1977,Nivat1978}]
Let $L\subseteq\Alphabet^{\omega}$.
$L\in\oCF \Leftrightarrow L=\wlang{\Cal G}$ for a MCFG $\Cal G$.
\end{theorem}

Now we have a converse of \cref{prop:ocf-is-munu-def} by:
\begin{proposition}
\label{prop:expr-to-mcfg}
    For each expression $e$ there is a MCFG $ \gram $ s.t.\ $\lang e = \lang \gram$.
\end{proposition}

\begin{proof}
[Proof sketch]
    Given a $\mu\nu$-expression $e$, we construct a grammar just like in \eqref{eq:can-grammar-mu-only}, but with extra clause $X_{\nu X f(X)} \to X_{f(\nu X f(X))}$. 
    We maintain two copies of each non-terminal, one \mgt{magenta} and one normal, so that a derivation also `guesses' a \mgt{$\nu$-thread} `on the fly'. Formally, the \mgt{magenta} productions of our grammar are:
    \[
    \begin{array}{r@{\quad \rightarrow \quad }l}
         \mgt{X_1} & \epsilon \\
         \mgt{X_a} & a \\
         \mgt{X_{f+g}} & \mgt{X_f} \quad | \quad \mgt{X_g} \\
         \mgt{X_{fg}} & \mgt{X_f}X_g \quad | \quad X_f\mgt{X_g} \\
         \mgt{X_{\mu X f(X)}}& \mgt{X_{f(\mu X f(X))}}\\
         \mgt{X_{\nu X f(X)}}& \mgt{X_{f(\nu X f(X))}}
    \end{array}
    \]
    Productions for normal non-terminals have only normal non-terminals on their RHSs.

    Now set $F$, the set of acceptable sets, to include all sets extending some $\{\mgt{X_f} : f \in E\}$, for $E$ with smallest expression a $\nu$-formula, by normal non-terminals.
    Now any accepting leftmost derivation of a word $w$ from $\mgt{X_e}$ describes a winning play of the evaluation puzzle from $(w,e)$ and vice-versa.
\end{proof}

\subsection{Proof search game and completeness}
In order to prove completeness of $\lmunuHKAnwf$, we need to introduce a game-theoretic mechanism for organising proof search, in particular so that we can rely on \emph{determinacy} principles thereof.

\begin{definition}
    [Proof search game]
    The \emph{proof search game} (for $\lmunuHKA$) is a two-player game played between Prover $(\prover)$, whose positions are inference steps of $\lmunuHKA$, and Denier $(\denier)$, whose positions are sequents of $\lmunuHKA$.
A \defname{play} of the game starts from a particular sequent:
at each turn, $\prover$ chooses an inference step with the current sequent as conclusion, and $\denier$ chooses a premiss of that step; the process repeats from this sequent as long as possible.

An infinite play of the game is \defname{won} by $\prover$ (aka \defname{lost} by $\denier$) if the branch constructed has a progressing thread; otherwise it is won by $\denier$ (aka lost by $\prover$). {In the case of deadlock, the player with no valid move loses.}
\end{definition}

\begin{proposition}
    [Determinacy ($\exists 0\#$)]
    \label{prop:det-prf-srch-game}
    The proof search game is determined, i.e.\ from any sequent $\Gamma \seqar S$, either $\prover $ or $\denier$ has a winning strategy.
\end{proposition}
Note that the winning condition of the proof search game is (lightface) analytic, i.e.\ $\Sigma^1_1$: ``there \emph{exists} a progressing thread''.
Lightface analytic determinacy lies beyond $\mathrm{ZFC}$, as indicated equivalent to the existence of $0\#$ \cite{Harrington1978}.
Further consideration of our metatheory is beyond the scope of this work.

\smallskip

It is not hard to see that $\prover$-winning-strategies are `just' $\infty$-proofs.
Our goal is to show a similar result for $\denier$, a sort of `countermodel construction'.

\begin{lemma}
\label{lem:d-strat-to-countermodel}
$\denier$ has a winning strategy from $\Gamma \seqar S$ $\implies$ $ \lang \Gamma \setminus \lang S \neq \emptyset$.
\end{lemma}

\noindent
Before proving this, let us point out that \cref{lem:invertibility} applies equally to the system $\munuHKA$.
We also have the useful observation:
\begin{proposition}
    [Modal]
    \label{lem:modal}
     $\lang {a\Gamma} \subseteq \{\epsilon\} \cup \bigcup\limits_{a\in \Alphabet}\lang{aS_a} \implies \lang{\Gamma} \subseteq \lang{S_a}$.
\end{proposition}
\noindent
This follows directly from the definition of $\lang\cdot$.
Now we can carry out our `countermodel construction' from $\denier$-winning-strategies:

\begin{proof}
[Proof of \cref{lem:d-strat-to-countermodel}]
Construct a $\prover$-strategy $\provstrat$ that is deadlock-free by always preserving validity, relying on \cref{lem:invertibility,lem:modal}.
In more detail $\provstrat$ does the following:
\begin{enumerate}
    \item\label{item:apply-leftmost-steps} Apply leftmost logical steps (on LHS or RHS) as long as possible.
    \item\label{item:lhs-empty-apply-id} If the LHS is empty and the RHS contains an empty list, then weaken the remainder of the RHS and apply $\id$.
    \item\label{item:lhs-nonempty-apply-modal} Otherwise, if the LHS has form $a\Gamma$ and RHS has form $aS_a,S$, where $S$ contains only lists that are empty or begin with some $b\neq a$, then apply $\wk$ and $\kl a $ to obtain the sequent $\Gamma \seqar S_a$ and go back to \ref{item:apply-leftmost-steps}.
\end{enumerate}
Now each iteration of \ref{item:apply-leftmost-steps} must terminate by left-guardedness and leftmostness.
This must end at a valid sequent, by \cref{lem:invertibility},  each of whose lists are either $\empty$ or begin with a letter, by inspection of the rules.
Now, if the LHS is empty, then the RHS must contain an empty list, and so step \ref{item:lhs-empty-apply-id} successfully terminates the preproof.
If the LHS has form $a\Gamma$, then step \ref{item:lhs-nonempty-apply-modal} applies and preserves validity by \cref{lem:modal}.
Note that any infinite play of $\provstrat$ must repeat step \ref{item:lhs-nonempty-apply-modal} infinitely often, as each iteration of \ref{item:apply-leftmost-steps} terminates, and so has infinitely many $\kl{}$ steps and is not ultimately stable.

Now, suppose $\denstrat$ is a $\denier$-winning-strategy and play $\provstrat$ against it to construct a play $B = (\sequent_i)_{i<\omega} = (  \Gamma_i \seqar S_i)_{i<\omega}$. Note that indeed this play must be infinite since (a) $\provstrat$ is deadlock-free; and (b) $\denstrat$ is $\denier$-winning.
Now, let $w = \prod\limits_{\kl a \in B} a$ be the product of labels of $\K$ steps along $B$, in the order they appear bottom-up. We claim $w \in \lang \Gamma \setminus \lang S$:
\begin{itemize}
    \item $w \in \lang \Gamma$. By construction $[\Gamma_i]_i$ has a subsequence forming an infinite play $\pi$ of the evaluation puzzle from $(w,\Gamma)$. 
    Since the play $B$ is won by $\denier$, $B$ cannot have a $\mu$-thread so it must have a $\nu$-thread (since it is i.o.p.), and so $\pi$ is winning. Thus $w \in \lang \Gamma$ by \cref{thm:evaluation-theorem}.
    \item $w \notin \lang S$. Take an arbitrary play $\pi$ of the evaluation puzzle from some $(w,\Delta)$  with $\Delta \in S$.
    This again induces an infinite sequence of cedents $[\Delta_i]_{i<\omega}$ along the RHSs of $B$.
    Now, $[\Delta_i]_{i<\omega}$ cannot have a  $\nu$-thread by assumption that $B$ is winning for $\denier$, and so $\pi$ is not a winning play of the evaluation puzzle from $(w,\Delta)$.
    Since the choices of $\Delta \in S$ and play $\pi$ were arbitrary, indeed we have $w\notin \lang S$ by \cref{thm:evaluation-theorem}.
\end{itemize}
\end{proof}

Now from \cref{prop:det-prf-srch-game,lem:d-strat-to-countermodel}, observing that $\prover$-winning-strategies are just $\infty$-proofs, we conclude:
\begin{theorem}[Completeness]
\label{thm:munuHKA-completeness}
$\lang \Gamma \subseteq \lang S \implies \lmunuHKA \nwfproves \Gamma \seqar S$.
\end{theorem}

\section{Complexity matters and further perspectives}

In this subsection we make further comments, in particular regarding the complexity of our systems, at the level of arithmetical and analytical hierarchies. 
These concepts are well-surveyed in standard textbooks, e.g.\ \cite{mansfield1985recursive,sacks2017}, as well as various online resources.

\medskip

\textbf{Complexity and irregularity for finite words.}
The equational theory of $\mu$-expressions in $\lang\cdot$ is actually $\Pi^0_1$-complete, i.e.\ co-recursively-enumerable, due to the same complexity of universality of context-free grammars (see, e.g., \cite{HopcroftUllmanMotwani06}).
In this sense there is no hope of attaining a finitely presentable (e.g.\ cyclic, inductive) system for the equational theory of $\mu$-expressions in $\lang\cdot$.
However it is not hard to see that our wellfounded system $\muHKAw$ enjoys optimal $\Pi^0_1$ proof search, thanks to invertibility and termination of the rules, along with decidability of membership checking.
Indeed a similar argument is used by Palka in \cite{Palka07} for the theory of `$*$-continuous action lattices'.
Furthermore let us point out that our non-wellfounded system also enjoys optimal proof search: $\muHKA \nwfproves \Gamma \seqar S$ is equivalent, by invertibility, to checking that \emph{every} sequent of form $\vec a \seqar S$ reachable by only left rules in bottom-up proof search has a polynomial-size proof (bound induced by length of leftmost derivations).
This is a $\Pi^0_1$ property.

\medskip

\textbf{Complexity and inaxiomatisability for infinite words.}
It would be natural to wonder whether a similar argument to \cref{sec:inf-ax-cf-theory} gives rise to some infinitary axiomatisation of the equational theory of $\mu\nu$-expressions in $\lang \cdot$. 
In fact, it turns out this is impossible: the equational theory of $\omega$-CFLs is $\Pi^1_2$-complete \cite{Finkel09}, so there is no hope of a $\Pi^0_1$ (or even $\Sigma^1_2$) axiomatisation.
In particular, the projection argument of \cref{sec:inf-ax-cf-theory} cannot be scaled to the full system $\lmunuHKA$ because $\cdot $ does not distribute over $\bigcap$ in $\lang\cdot$, for the corresponding putative `right $\omega $ steps' for $\nu$.
For instance $0 = ((aa)^* \cap a(aa)^*) a^* \neq (aa)^*a^* \cap a(aa)^*a^* = aa^*$.
Indeed let us point out that here it is crucial to use our hypersequential system $\HKA$ as a base rather than, say, the intuitionistic systems of other proof theoretic works for regular expressions (and friends) \cite{Palka07,DasPous18:lka-pt}: the appropriate extension of those systems by $\mu$s and $\nu$s should indeed enjoy an $\omega$-translation, due to only one formula on the right, rendering them incomplete.

Again let us point out that $\infty$-provability in $\lmunuHKA$, in a sense, enjoys optimal complexity.
By determinacy of the proof search game, $\lmunuHKA \nwfproves \Gamma \seqar S $ if and only if there is \emph{no} $\denier$-winning-strategy from $\Gamma\seqar S$.
The latter is indeed a $\Pi^1_2$ statement: ``\emph{for every} $\denier$-strategy, \emph{there exists} a play along which \emph{there exists} a progressing thread''.

\medskip

\textbf{Comparison to \cite{GraHenKoz13:inf-ax-cf-langs}.}
Our method for showing completeness of $\muHKAw$ is quite different from the analogous result of \cite{GraHenKoz13:inf-ax-cf-langs} which uses the notion of `rank' for $\mu$-formulas, cf.~\cite{AlbKraStu14:just-ind-mu-forms}.
Our result is somewhat stronger, giving \emph{cut-free} completeness, but it could be possible to use ranks directly to obtain such a result too. 
More interestingly, the notion of projections and $\omega$-translation should be well-defined (for LHS $\mu$ formulas) even in the presence of $\nu$s, cf.~\cite{DDS23}, whereas the rank method apparently breaks down in such extensions. 
This means that our method should also scale to $\munuHKA$ $\infty$-proofs where, say, each infinite branch has a LHS $\mu$-thread.
It would be interesting to see if this method can be used to axiomatise some natural fragments of $\omega$-context-free inclusions.

Note that, strictly speaking, our completeness result for $\CFA$ was only given for the guarded fragment. However it is known that $\CFA$ (and even weaker theories) already proves the equivalence of each expression to one that is even left-guarded, by formalising conversion to Greibach normal form \cite{EsikLeiss05:alg-comp-semirings}.

\section{Conclusions}
In this work we investigated of the proof theory of context-free languages (CFLs) over a syntax of $\mu$-expressions. 
We defined a non-wellfounded proof system $\muHKAnwf$ and showed its soundness and completeness for the model $\lang\cdot$ of context-free languages.
We used this completeness result to recover the same for a cut-free $\omega$-branching system $\muHKAw$ via proof-theoretic techniques. 
This gave an alternative proof of the completeness for the theory of $\mu$-continuous Chomsky algebras from \cite{GraHenKoz13:inf-ax-cf-langs}.
We extended $\mu$-expressions by \emph{greatest} fixed points to obtain a syntax for $\omega$-context-free languages. We studied an extension by \emph{greatest} fixed points, $\lmunuHKAnwf$ and showed its soundness and completeness for the model $\lang\cdot$ of context-free languages, employing game theoretic techniques.

\smallskip

Since inclusion of CFLs is $\Pi^0_1$-complete, no recursively enumerable (r.e.) system can be sound and complete for their equational theory. 
However, by restricting products to a letter on the left one can obtain a syntax for \emph{right-linear grammars}. 
Indeed, for such a restriction complete cyclic systems can be duly obtained~\cite{DasDe24}. 
It would be interesting to investigate systems for related decidable or r.e.\ inclusion problems, e.g.\ inclusions of context-free languages in regular languages, and inclusions of \emph{visibly pushdown} languages~\cite{AlurMadhusudan04,AlurMadhusudan09}.

\smallskip

The positions of our evaluation puzzle for $\mu\nu$-expressions use cedents to decompose products, similar to the stack of a pushdown automaton, rather than requiring an additional player. 
Previous works have similarly proposed model-checking games for (fragments/variations of) context-free expressions, cf.~\cite{Lange02:model-checking-fp-logic+chop,LoedMadhSerr04:vp-games}, where more complex winning conditions seem to be required.
It would be interesting to compare our evaluation puzzle to those games in more detail.

Note
that our completeness result, via determinacy of the proof search game, depends on the assumption of (lightface) analytic determinacy. It is natural to ask whether this is necessary, but this consideration is beyond the scope of this work. 
Let us point out, however, 
that even $\omega$-context-free determinacy exceeds the capacity of $\mathrm{ZFC}$~\cite{Finkel:determinacy,LiTanaka17}.

\smallskip

Finally, it would be interesting to study the \emph{structural} proof theory arising from systems $\muHKAnwf$ and $\munuHKAnwf$, cf.~\cite{DasPous18:lka-pt}. 
It would also be interesting to see if the restriction to leftmost $\infty$-proofs can be replaced by stronger progress conditions, such as the `alternating threads' from \cite{DasGirlando22,DasGirlando23}, in a similar hypersequential system for predicate logic.
Note that the same leftmost constraint was employed in \cite{HazKup22:transfin-HKA} for an extension of $\HKA$ to \emph{$\omega$-regular languages}.

\bibliographystyle{alpha}
\bibliography{biblio}

\end{document}